\documentclass[draftclsnofoot,onecolumn,12pt]{IEEEtran}
\ifCLASSINFOpdf
\else
\fi
\usepackage{graphics}
\usepackage{graphicx}
\usepackage{epsfig}
\usepackage{subfigure}
\usepackage{amssymb}
\usepackage{amsmath}
\usepackage{amsfonts}
\usepackage{stmaryrd}
\usepackage{txfonts}
\usepackage{mathrsfs}
\usepackage{latexsym, bm}
\newtheorem{theorem}{Theorem}
\newtheorem{proposition}{Proposition}

\newtheorem{lemma}{Lemma}
\newtheorem{corollary}{Corollary}

\newtheorem{remark}{Remark}

\newtheorem{proof}{Proof}
\newcommand*{\QEDA}{\hfill\ensuremath{\blacksquare}}

\begin{document}
%

\title{Consensus of Multi-agent Systems Under State-dependent Information Transmission}
%
%
%

\author{Gangshan~Jing ,~
       Yuanshi~Zheng,~
       and~Long~Wang~
\thanks{This work was supported by NSFC (Grant
Nos. 61020106005, 61375120 and 61304160) and the Fundamental
Research Funds for the Central Universities (Grant Nos. JB140406).}
\thanks{G. Jing and Y. Zheng are with Key Laboratory of Electronic Equipment Structure Design of Ministry of Education, School of Mechano-electronic Engineering, Xidian University, Xi'an 710071, China and also with the Center for Complex Systems, School of Mechano-electronic Engineering, Xidian University, Xi'an 710071, China
 (e-mail:nameisjing@gmail.com;zhengyuanshi2005@163.com)
 }
\thanks{L. Wang is with the Center for Systems and Control, College of Engineering, Peking
University, Beijing 100871, China
 (e-mail: longwang@pku.edu.cn)}
}

\maketitle

\begin{abstract}
In this paper, we study the consensus problem for continuous-time and discrete-time multi-agent systems in state-dependent switching networks. In each case, we first consider the networks with fixed connectivity, in which the communication between adjacent agents always exists but the influence could possibly become negligible if the transmission distance is long enough. It is obtained that consensus can be reached under a restriction of either the decaying rate of the transmission weight or the initial states of the agents. After then we investigate the networks with state-dependent connectivity, in which the information transmission between adjacent agents gradually vanishes if their distance exceeds a fixed range. In such networks, we prove that the realization of consensus requires the validity of some initial conditions. Finally, the conclusions are applied to models with the transmission law of C-S model, opinion dynamics and the rendezvous problem, the corresponding simulations are also presented.
\end{abstract}

\begin{keywords}
Multi-agent systems, state-dependent, switching networks, opinion dynamics, rendezvous.
\end{keywords}

%
\IEEEpeerreviewmaketitle

\section{Introduction}

\IEEEPARstart{D}{istributed} cooperative control of systems with multiple agents has attracted attention from different research communities in recent several years. In these systems, all the agents interact with each other via a communication topology and only local information can be employed. Therefore, in order to drive them to accomplish tasks, a distributed control law is required. A multi-agent system has a wide range of applications since it can perform a variety of collective behaviors. For instance, formation of unmanned aerial vehicles \cite{Xiao09}, attitude adjustment of spacecrafts \cite{RenA07}, flocking of multiple robots \cite{Jia13} and so on. During these challenging topics, reaching consensus is a crucial problem that we have to deal with. Moreover, many collective behaviors can be performed based on strategies to reach consensus.

So far, there have been numerous references related to the consensus problem. More specifically, \cite{Saber04} considered the consensus of continuous-time systems in which agents are of single integrator dynamics, the authors found the connectivity of the network plays an important role in reaching consensus. On this basis, the static and dynamic consensus protocols for continuous-time systems with double integrator dynamics are studied in \cite{Xie07} and \cite{Ren07}, respectively. In \cite{Zheng11,Zheng12}, the authors investigated the consensus of a heterogeneous system which consists of a number of agents with single and double integrator dynamics simultaneously. For discrete-time systems, \cite{Ren05} investigated the first-order multi-agent systems and obtained a necessary and sufficient condition for consensus. All these works also considered the case of time-dependent switching networks. It was shown that by employing their protocols, if the communication topology switches in a finite number of connected graphs, the conclusion for consensus still holds. Moreover, some literatures also have studied consensus in time-dependent networks in depth \cite{Moreau05,Cao11,Hendrickx13,Su13,Zheng14}.

In fact, there exist many systems running in switching networks which are closely related to the states of agents. As an example, in Vicsek's model \cite{Vicsek95}, all the agents keep the same speed but different headings, the key to realize swarming is making each agent update its heading by averaging the headings of agents who are close to it. For these systems, the information transmission weight varies when the agents change their states and thus there may exist an infinite number of communication graphs to be employed. Furthermore, with the evolution of the system, the connectivity of the communication topology can be possibly broken, which will lead to the failure of consensus. Therefore, such systems have very different properties and are worth exploring. A few investigations have been carried out on this issue. Cucker and Smale proposed a flocking model(C-S model) via a transmission weight dependent on state distance in \cite{Cucker07,Cucker10}. The communication weight is designed like gravity, $i.e.$, as the distance between two agents increases, the information they receive from each other gradually weakens but always exists. This implies that the communication topology is always a complete graph. The authors' research shows that convergence can be achieved under a restriction on the initial states, which is really different with the previous results of systems in time-dependent switching networks. Besides, the model of opinion dynamics introduced by Hegselmann and Krause in \cite{Krause02} is also an interesting topic. It describes the evolution of a number of opinions in a group of agents who can interact with their neighbors. Different from C-S model, H-K(Hegselmann-Krause) model includes a bounded confidence constraint, so that each agent can only interact with the agents who keep opinions within the confidence bound of its opinion. Therefore, H-K model allows both the addition and loss of links in communication topology, and thus the connectivity cannot be always kept. Several literatures related to opinion dynamics have been conducted \cite{Krause00,Blondel09,Blondel10,Ceragioli12}. In \cite{Yang14}, the author obtained a sufficient condition for consensus of continuous-time opinion model by maintaining the distance between any two agents nonincreasing. Similar to opinion dynamics, the rendezvous problem of multi-agent systems also involves the the uncertainty of the network's connectivity \cite{Lin071,Lin072}. In order to realize rendezvous, \cite{Su10} proposed an algorithm by employing a potential function to preserve the network's connectivity. Also the information transmitted between agents in \cite{Li13} and \cite{Andreasson13} is influenced by the agents' states.

Out of the above-mentioned situation, we consider the consensus problem of multi-agent systems with a general state-dependent information transmission weight. Two kinds of state-dependent switching networks are considered. In the first case, switching has no effects on the connectivity of the communication topology. Different from \cite{Andreasson13}, we mainly explore systems with damping information transmission weight without coupling extra nonlinear gains, and the communication between agents is only affected by their relative states. That is, the transmission mode in our study contains the one of C-S model as a special case. In the second case, the communication graph is fully dependent on the states of all the agents. The connectivity of the communication topology can be varying as the system evolves. Hence it can apply to opinion dynamics and the rendezvous problem. In this paper, we always assume the influence between agents decays as their distance increases. This assumption can be taken off in several circumstances, we will state it in the text.

In this paper, we investigate the consensus problem of continuous-time and discrete-time multi-agent systems respectively. For each kind of the systems, agents with first-order and second-order dynamics are separately considered. The corresponding protocols are proposed by employing protocols in the previous literatures with state-dependent communication weight instead. By using Lyapunov method and reduction to absurdity, a sufficient condition to consensus for each protocol is obtained. We find that for a part of systems with the first kind of weight, consensus can be reached under a restriction of initial states. And that for all the systems with the second kind of weight, we always require the agents' initial states to satisfy a condition for reaching consensus. Finally, we apply our results to C-S model, opinion dynamics and the rendezvous problem. Some simulations are performed to illustrate the effectiveness of the theoretical results.

{\bf Notation:} Throughout this paper, we denote the set of real numbers by $\mathbb{R}$, the set of positive real numbers by $\mathbb{R}_{>0}$, and the set of nonnegative real numbers by $\mathbb{R}_{\geq0}$. Let $\mathbb{R}^n$ be the $n-$dimensional Euclidean space, $||\cdot||$ be the Euclidean norm. $X^T$ stands for the transpose of matrix $X$, $|\mathcal{V}|$ is the cardinality of set $\mathcal{V}$. $H_0(A)$ denotes the eigenspace of matrix $A$ corresponding to zero. $\pi_{M}(x)$ denote the orthogonal projection of $x$ onto space $M$. $dim(M)$ is the dimension of space $M$. $\otimes$ represents the kronecker product. For a matrix $A\in \mathbb{R}^{n\times n}$, $\lambda_i(A)$ denotes the $i$th eigenvalue of $A$, $i.e.$, $\lambda_1(A)\leq\cdots\leq\lambda_n(A)$. $\lfloor x\rfloor$ is the largest integer not greater than $x$ and $\lceil x\rceil$ is the smallest integer not less than $x$.

\section{Problem formulation}

\subsection{Preliminaries of Graph Theory}
We use a graph $\mathcal{G}=(\mathcal{V},\mathcal{E},\mathcal{A})$ to denote the communication relationship between agents. $\mathcal{V}$ is a set consisting of some vertices, each vertex corresponds to an agent in the system. $\mathcal{E}$ is the set of edges, each edge is denoted by a pair of agents, $i.e.$, $(i,j)$. In this paper, we propose a matrix $G=[G_{ij}]\in \mathbb{R}^{n\times n}$ to show the distribution of communication links in the network. That is, $G_{ij}=1$ if $(i,j)\in\mathcal{E}$, and $G_{ij}=0$ otherwise. The set of neighbors of agent $i$ is denoted by $\mathcal{N}_i=\{j~|~G_{ij}=1\}$. $\mathcal{A}=[a_{ij}]\in\mathbb{R}^{n\times n}$ is a matrix describing the weight of information flow between agents, in which $a_{ij}$ denotes the information transmission weight between agents $i$ and $j$. Throughout this paper, we always assume that $\mathcal{G}$ is undirected, which implies that both $G$ and $\mathcal{A}$ are symmetric matrices. We use a diagonal matrix $\Delta=[\Delta_{ij}]$ with $\Delta_{ii}=\sum_{j\in\mathcal{V}}a_{ij}$ to show the degree of each agent, the Laplacian matrix of graph $\mathcal{G}$ is defined by $L=\Delta-\mathcal{A}$. By Gerschgorin Theorem, it can be easily proved that $L$ is a positive semi-definite matrix. In our work, the communications between agents may be always changing as the agents' states evolve. Hence we use $L_x$ to denote the Laplacian matrix according to state $x$ for continuous-time systems, and $L_t$ to denote the Laplacian matrix at step $t$ for discrete-time systems. A path between $i$ and $j$ in graph $\mathcal{G}$ is a sequence of distinct edges of the form $(i_1,i_2)$, $(i_2,i_3)$, $\cdots$, $(i_{i-1},i_k)$, where $i_1=i$, $i_k=j$, and $(i_r,i_{r+1})\in\mathcal{E}$ for $r\in\{1,\cdots,k-1\}$. A graph is said to be connected if there exists a path between any two distinct vertices of the graph.

The connectivity of graph $\mathcal{G}$ is written by $\kappa(\mathcal{G})$, which is the minimum size of a vertex set $S$ such that $\mathcal{G}-S$ is disconnected or has only one vertex. Therefore, $\kappa(\mathcal{G})$ can be confirmed only by $G$. Furthermore, it is straightforward to see that $\kappa(\mathcal{G})>0$ if and only if $\mathcal{G}$ is connected. Given $i,j\in\mathcal{V}(\mathcal{G})$, a set $S\subseteq\mathcal{V}(\mathcal{G})-\{i,j\}$ is an $i,j-$cut if $\mathcal{G}-S$ has no paths between $i$ and $j$.

\subsection{Systems and Consensus}
For continuous-time systems, we consider agents with both single integrator dynamics
\begin{equation}\label{ct1}
\dot{x}_i=u_i, ~~~~i\in \mathcal{V}
\end{equation}
and double integrator dynamics
\begin{equation}\label{ct2}
\begin{split}
&\dot{x}_i=v_i,\\
&\dot{v}_i=u_i, ~~~~i\in \mathcal{V}.
\end{split}
\end{equation}
For discrete-time systems, agents with both first-order dynamics
\begin{equation}\label{dt1}
x_i(t+1)=x_i(t)+u_i(t), ~~~~i\in\mathcal{V}
\end{equation}
and second-order dynamics
\begin{equation}\label{dt2}
\begin{split}
&x_i(t+1)=x_i(t)+k_1v_i(t),\\
&v_i(t+1)=v_i(t)+u_i(t),  ~~~~i\in\mathcal{V}
\end{split}
\end{equation}
are considered.

In this paper, we suppose $k_1>0$, $\mathcal{V}=\{1,\cdots,n\}$, $x_i,~v_i,~u_i\in\mathbb{R}^m$, where $m$ is a positive integer. Let $\mathbb{E}=\mathbb{R}^m$, then $x=(x_1^T,\cdots,x_n^T)^T, ~v=(v_1^T,\cdots,v_n^T)^T\in\mathbb{E}^n$. In the following, a matrix in $\mathbb{R}^{n\times n}$ may act on $\mathbb{E}^n$. That is, $Ax=(A\otimes I_m)x$ for $A\in\mathbb{R}^{n\times n}$, $x\in\mathbb{E}^n$. We say the consensus problem is solved if $x$ gradually evolves into $M=span\{\mathbf{1}_n\otimes r~|~r\in\mathbb{E}\}$ as $t\rightarrow\infty$. Specifically, if $M=\{\mathbf{1}_n\otimes\frac{1}{n}\sum_{i\in\mathcal{V}}x_i(0)\}$, the average consensus is said to be solved. Let $e_i$, $i=1,\cdots,m$ be the standard orthogonal basis of $\mathbb{R}^m$, $i.e.$, $e_i=(0_{(i-1)},1,0_{(m-i)})^T$. Then $f_i=\frac{1}{\sqrt{n}}\mathbf{1}_n\otimes e_i$, $i=1,\cdots,m$ are the standard orthogonal basis of $M$. Therefore, the orthogonal projection of $x$ onto $M$ is
\[
\begin{split}
\pi_M(x)&=\sum_{i=1}^m\langle x,f_i\rangle f_i\\
&=\sum_{i=1}^m\langle x,\frac{1}{\sqrt{n}}\mathbf{1}_n\otimes e_i\rangle\cdot \frac{1}{\sqrt{n}}\mathbf{1}_n\otimes e_i\\ &=\mathbf{1}_n\otimes\frac{1}{n}\sum_{i\in\mathcal{V}}x_i.
\end{split}
\]
For convenience of the proofs, we set $p=x-\pi_M(x)$ and $q=v-\pi_M(v)$. Hence, consensus is reached if and only if $p\rightarrow0$ and $q\rightarrow0$ as $t\rightarrow\infty$.

\subsection{Useful Lemmas}

For convenience in the proofs of the main results, several lemmas associated with graphs and matrices are listed below.
\begin{lemma}\label{rankL}
If graph $\mathcal{G}=(\mathcal{V},\mathcal{E},\mathcal{A})$ with $\mathcal{V}=\{1,\cdots,n\}$ is connected, then $H_0(L\otimes I_m)=span\{\mathbf{1}_n\otimes r~|~r\in\mathbb{E}\}=M$, where $L$ is the Laplacian matrix of $\mathcal{G}$.
\end{lemma}
\begin{lemma}(\cite{Cortes06})\label{eigenvalue}
Given a positive semi-definite $d\times d$ matrix $A$, we have $x^TAx\geq\lambda_2(A)||x-\pi_{H_0(A)}(x)||^2$, for any $x\in\mathbb{R}^d$.
\end{lemma}
\begin{lemma}(\cite{Cucker10})\label{||x||}
For all $x\in\mathbb{E}^n$, $L\in\mathbb{R}^n$ is the Laplacian matrix of a graph, we have:

$(1)$ $||x_i-x_j||=||p_i-p_j||\leq\sqrt{2}||p||$;

$(2)$ $\displaystyle{\frac{1}{2n}\sum_{i\in\mathcal{V}} \sum_{j\in\mathcal{V}} ||x_i-x_j||^2= \frac{1}{2n}\sum_{i\in\mathcal{V}} \sum_{j\in\mathcal{V}} ||p_i-p_j||^2=||p||^2}$;

$(3)$ $x^TLx=\langle x,Lx\rangle=\displaystyle{ \frac{1}{2} \sum_{i\in\mathcal{V}} \sum_{j\in\mathcal{V}}}a_{ij}(x)||x_i-x_j||^2\geq0$.
\end{lemma}
\begin{lemma}\label{graph}
Suppose that the connectivity of graph $\mathcal{G}$ is $\kappa(\mathcal{G})=k^*>0$, then there exist at least $k^*$ disjoint paths between any different vertices.
\end{lemma}
\begin{lemma}\label{n-1}
If graph $\mathcal{G}$ is not connected, then there exist at least $n-1$ pairs of disconnected nodes in the graph.
\end{lemma}
The relevant proofs will be stated in Section 7.

\section{Consensus of Continuous-time Multi-agent Systems}

The consensus problem of continuous-time multi-agent systems has been studied in many previous works. In this section, we employ the consensus protocols widely used before and assume the information transmission between the agents becomes state-dependent. It will be shown that a very different result emerges due to this change.

\subsection{Continuous State-dependent Transmission Weight}
We consider two classes of systems with state-dependent information transmission. The first case is of fixed connectivity in communication topology, which implies that $G$ and $\kappa(\mathcal{G})$ are invariant. The communication weight between agents $i$ and $j$ is set as $a_{ij}=G_{ij}\alpha(||x_i-x_j||^2)$, where $\alpha(s)$ is a positive function which decays as the increasing of $s$. Therefore, for agent $i$, the information that it receives from agent $j$ can be denoted by $G_{ij}\alpha(||x_i-x_j||^2)(x_i-x_j)$.
We have the following assumption on $\alpha(\cdot)$.

\emph{Assumption 1}: $\alpha(\cdot): \mathbb{R}_{\geq0}\rightarrow \mathbb{R}_{>0}$ is continuous and nonincreasing, $\alpha(0)<\infty$.

In the second case, the connectivity of communication graph $\mathcal{G}=(\mathcal{V}, \mathcal{E}, \mathcal{A})$ is entirely dependent on the states of all the agents. More specifically, the communication weight between $i$ and $j$ is $a_{ij}=G_{ij}\alpha(||x_i-x_j||^2)= \alpha(||x_i-x_j||^2)$, because $G_{ij}=1$ if and only if $\alpha(||x_i-x_j||^2)\neq0$. $\alpha(\cdot)$ is under the following assumption.

\emph{Assumption 2}: $\alpha(\cdot): \mathbb{R}_{\geq0}\rightarrow \mathbb{R}_{\geq0}$ is continuous and nonincreasing, $\alpha(0)<\infty$, $\alpha(s)>0$ if $s<R^2$, $\alpha(s)=0$ if $s\geq R^2$, where $R\in\mathbb{R}_{>0}$ is a constant.

For simplicity, we denote $\alpha(||x_i-x_j||^2)$ by $\alpha_{ij}(x)$ in the rest of the paper.

We study continuous-time systems in this section. Assumption 1 and Assumption 2 will be performed respectively. It is shown that when the nonlinear weight is coupled with the state difference, a number of characteristics of these systems will emerge.

\subsection{Consensus with Fixed Connectivity of Networks}

In the case of fixed connectivity, a very long distance between a pair of agents may cause their information transmission becoming slight and cannot work effectively. For reaching consensus, we hope to obtain a bound of the distance between any agents. In the results, we will see that the boundedness of $||p||$ is the key to solve the consensus problem. Once $||p||$ is guaranteed to be bounded, the following lemma shows that the algebraic connectivity of the communication graph, written by $\lambda_2(L)$, has a nonzero lower bound. The corresponding proof is presented in Section 7.
\begin{lemma}\label{bound}
Under Assumption 1. For any $t\geq0$, if $||p(t)||$ is upper bounded, and the communication topology is connected, then $\lambda_2(L_x)$ has a nonzero lower bound.
\end{lemma}
Consider a group of agents with dynamics (\ref{ct1}), the protocol in \cite{Saber04} is studied:
\begin{equation}\label{cu1}
u_i=\sum\limits_{j=1}^nG_{ij} \alpha_{ij}(x)(x_j-x_i).
\end{equation}
\begin{theorem}\label{th cu1}
Consider a system consisting of $n$ agents with dynamics (\ref{ct1}). Under Assumption 1, protocol (\ref{cu1}) globally asymptotically solves the average consensus problem if the communication topology is connected.
\end{theorem}
\begin{proof}It is easy to see that $p$ satisfies the same differential function as $x$ does. Consider the Lyapunov function $V(p)=\frac{1}{2}||p||^2$, the positive definiteness of $V(p)$ obviously holds. Since graph $\mathcal{G}$ is connected, together with Lemma \ref{rankL}, it follows that $H_0(L_x\otimes I_m)=M$. According to Lemma \ref{eigenvalue}, we have $\dot{V}=-p^TL_xp\leq-\lambda_2(L_x\otimes I_m)||p-\pi_M(p)||^2=-\lambda_2(L_x)||p||^2$. Thus, $V(p(t))\leq V(p(0))$ for any $t\geq0$, implying that $||p||$ is bounded by $||p(0)||$. From Lemma \ref{bound}, there exists a constant $c>0$, such that $\lambda_2(L_x)\geq c$. Consequently, $\dot{V}\leq-c||p||^2$. That is, $\dot{V}$ is negative definite. Together with the radial unboundedness of $V$, $p$ globally asymptotically converge to $0$. Due to the fact that graph $\mathcal{G}$ is undirected, we have the symmetry of $L_x$, then $\sum_{i\in\mathcal{V}}\dot{x}_i(t)=0$ for $t\geq0$. Therefore, $\lim\limits_{t\rightarrow\infty}x_i(t)= \frac{1}{n}\sum_{i\in\mathcal{V}}x_i(0)$ for any $i\in\mathcal{V}$. That is, all the agents globally asymptotically achieve the average consensus.
\end{proof}

For agents with dynamics (\ref{ct2}), we first study the static consensus protocol in \cite{Xie07}:
\begin{equation}\label{cu21}
u_i=-kv_i+ \sum\limits_{j=1}^nG_{ij}\alpha_{ij}(x)(x_j-x_i),
\end{equation}
where $k>0$ is the feedback gain of agent $i$.
\begin{theorem}\label{th cu21}
Consider a system consisting of $n$ agents with dynamics (\ref{ct2}). Under Assumption 1, protocol (\ref{cu21}) globally asymptotically solves the consensus problem if the communication topology is connected. Specifically, if the sum of the initial velocity of each agent is zero, the average consensus problem is solved.
\end{theorem}
\begin{proof} Consider the following Lyapunov function candidate.
\begin{equation}\label{lyapunov cu21}
V(x,v)=||kx+v||^2+||v||^2+\sum\limits_{i\in\mathcal{V}} \sum\limits_{j\in\mathcal{V}}\int_{0}^{||x_i-x_j||^2} a_{ij}(s)ds,
\end{equation}
where $a_{ij}(s)=G_{ij}\alpha(s)$. Let $||x||^2+||v||^2\rightarrow \infty$, one has $\sqrt{V(x,v)}\geq||kx+v||$, and $\sqrt{V(x,v)}\geq||v||$. Then $3\sqrt{V(x,v)}\geq||kx||-||v||+2||v||= ||kx||+||v||$, it follows that
$$
V(x,v)\geq\frac{1}{9}(||kx||+||v||)^2\geq \frac{1}{9}\min\{k^2,1\}(||x||^2+||v||^2) \rightarrow\infty.
$$
Then the radial unboundedness of $V(x,v)$ follows. The derivative of $V(x,v)$ along the trajectories of the agents is given by
\[
\begin{split}
\dot{V}&=2(kx+v)^T(kv-kv-L_xx)+2v^T(-kx-L_xx)\\
&~~~~+2\sum\limits_{i\in\mathcal{V}} \sum\limits_{j\in\mathcal{V}}G_{ij}\alpha_{ij}(x)(x_i-x_j)^T(v_i-v_j)\\
&=-2kx^TL_xx-2kv^Tv\leq0.
\end{split}
\]
Therefore, $\Omega(x,v)= \{x,v~|~V(x(t),v(t))\leq V(x(0),v(0))\}$ is positively invariant. Since $V$ is continuous, $V^{-1}[0,V(x(0),v(0))]$ is closed. Together with the radial unboundedness of $V$, $\Omega$ is bounded and thus is compact. By employing LaSalle's invariance principle, $\dot{V}\rightarrow0$ as $t\rightarrow\infty$, and since graph $\mathcal{G}$ is connected, together with Lemma \ref{rankL}, $x$ will evolve into $M$, and $v\rightarrow0$ as $t\rightarrow\infty$. That is, the position states of all the agents globally asymptotically reach consensus and the velocity of them vanish to zero in the end.

Moreover, let $U(x,v)= \sum_{i\in\mathcal{V}}v_i+ k\sum_{i\in\mathcal{V}}x_i$. Then $\dot{U}=-k\sum_{i\in\mathcal{V}}v+ \sum_{i\in\mathcal{V}}\sum_{j\in\mathcal{V}}(x_j-x_i) +k\sum_{i\in\mathcal{V}}v=0$. That is, $U(x^*,v^*)= U(x(0),v(0))$, where $x^*$ is the consensus position state of each agent. Therefore, it can be obtained that $\displaystyle{x^*= \frac{\sum_{i\in\mathcal{V}}v_i(0)+ k\sum_{i\in\mathcal{V}}x_i(0)}{nk}}$.

If $\sum_{i\in\mathcal{V}}v_i(0)=0$, it is easy to obtain $x^*=\frac{1}{n}\sum_{i\in\mathcal{V}}x_i$, which implies that the average consensus is achieved.
\end{proof}

Now we consider the dynamic consensus protocol proposed in \cite{Ren07}:
\begin{equation}\label{cu22}
u_i=\sum\limits_{j=1}^nG_{ij} \alpha_{ij}(x)(v_j-v_i)+ \sum\limits_{j=1}^nG_{ij} \alpha_{ij}(x)(x_j-x_i).
\end{equation}
Protocol (\ref{cu21}) makes the velocity of each agent gradually vanish to zero for arbitrary initial value, and thus always keeps the distance between any two agents constant in the steady state even if consensus is not reached. Hence, the compactness of $\Omega$ can be unconditionally guaranteed, and note that $||p||$ is also bounded. However, each agent applying protocol (\ref{cu22}) may obtain a nonzero velocity in the steady state, the distance between agents may be unbounded ($||p||$ will also be unbounded). To achieve global convergence, a condition of $\alpha(\cdot)$ is required to be appended.
\begin{theorem}\label{th cu22}
Consider a system consisting of $n$ agents with dynamics (\ref{ct2}). Under Assumption 1, suppose $\int_0^{\infty}\alpha(s)ds=\infty$, protocol (\ref{cu22}) globally asymptotically solves the consensus problem if the communication topology is connected.
\end{theorem}
\begin{proof} It is clear that $x$ and $v$ in system (\ref{ct2}) with (\ref{cu22}) can be replaced by $p$ and $q$. Consider the following energy-like function
\begin{equation}\label{lyapunov cu22}
V(p,q)=||q||^2+\frac{1}{2}\sum\limits_{i\in\mathcal{V}} \sum\limits_{j\in\mathcal{V}}\int_{0}^{||p_i-p_j||^2} a_{ij}(s)ds.
\end{equation}
Differentiating $V(p,q)$ along the trajectories of agents, one has
\[
\begin{split}
\dot{V}(p,q)&=2q^T(-L_xp-L_xq)+\sum\limits_{i\in\mathcal{V}} \sum\limits_{j\in\mathcal{V}}G_{ij}\alpha_{ij}(p)(p_i-p_j)^T(q_i-q_j) \\
&=-2q^TL_xq\leq0.
\end{split}
\]
Then the set $\Omega=\{p,q| V(p,q)\leq V(p(0),q(0))\}$ is positively invariant. Before employing LaSalle's invariance principle, it is necessary to prove the compactness of $\Omega$. It is clear that $||q(t)||$ is bounded by $V(p(0),q(0))$ for any $t\geq0$. Suppose $||p(t)|| \rightarrow \infty$ as $t\rightarrow t^*$, $t^*>0$ ($t^*$ can be infinite). From Lemma \ref{||x||}, there exist a pair of agents $i$ and $j$, such that $||x_i-x_j||\rightarrow\infty$ as $t\rightarrow t^*$. Since the communication graph is connected, there exists a path $(i,i_1),...,(i_s,j)$ between $i$ and $j$. Note that $||x_i-x_j||\leq||x_i-x_{i_1}||+\cdots+||x_{i_s}-x_j||$. Therefore, there exists a constant $k\in\{1,\cdots,s\}$, such that $||x_{i_k}-x_{i_{k+1}}||\rightarrow\infty$. This yields
$$
\sum\limits_{i\in\mathcal{V}} \sum\limits_{j\in\mathcal{V}}\int_{0}^{||p_i-p_j||^2} G_{ij}\alpha(s)ds \geq \int_0^{||x_{i_k}-x_{i_{k+1}}||^2}\alpha(s)ds\rightarrow\infty,
$$
as $t\rightarrow t^*$, which conflicts with $V(p,q)\leq V(p(0),q(0))$ for all $t\geq0$. Thus, $||p(t)||$ is always bounded. Together with $||q||^2\leq V(p,q)\leq V(p(0),q(0))$, it follows the radial unboundedness of $V(p,q)$ and the compactness of $\Omega$. Therefore, all the solutions of system (\ref{ct2}) with protocol (\ref{cu22}) globally asymptotically converge into the largest invariant set in $\{\dot{V}(p,q)=0\}$. From the connectivity of the communication graph and Lemma \ref{rankL}, both $p$ and $q$ will evolve into $M$. That is, $p_i-p_j\rightarrow0$, $q_i-q_j\rightarrow0$, as $t\rightarrow\infty$ for any $i,j\in \mathcal{V}$. Note that $\sum_{i\in\mathcal{V}}p_i=\sum_{i\in\mathcal{V}}q_i=0$, therefore, $p\rightarrow0$, $q\rightarrow0$, as $t\rightarrow\infty$. That is, all the agents globally asymptotically achieve consensus.
\end{proof}

The restriction of $\alpha(\cdot)$ is actually for the decaying rate of the communication. It is clear that the faster $\alpha(\cdot)$ damps, the more difficult the condition is satisfied. When $\int_0^{\infty}\alpha(s)ds=\infty$ is false, protocol (\ref{cu21}) solves the consensus problem if the initial states of all the agents are restricted. The following corollary states it in detail.
\begin{corollary}\label{co cu22}
Consider a system consisting of $n$ agents with dynamics (\ref{ct2}). Under Assumption 1, suppose $\int_0^{\infty}\alpha(s)ds<\infty$, the communication graph $\mathcal{G}$ is connected and the following inequality holds.
\begin{equation}\label{inequality cu22}
||q(0)||^2+ \frac{1}{2}\sum\limits_{i\in\mathcal{V}} \sum\limits_{j\in\mathcal{V}}\int_0^{||p_i(0)-p_j(0)||^2}G_{ij}\alpha(s)ds< k^*\int_0^{\infty}\alpha(s)ds,
\end{equation}
where $k^*$ is the connectivity of graph $\mathcal{G}$. Then protocol (\ref{cu22}) solves the consensus problem asymptotically.
\end{corollary}
\begin{proof} We still consider the energy-like function (\ref{lyapunov cu22}), the next step is to show the compactness of $\Omega=\{p,q| V(p,q)\leq V(p(0),q(0))\}$. Suppose $||p||\rightarrow\infty$, then there exist a pair of agents $i$ and $j$, such that $||p_{i}-p_j||\rightarrow\infty$. By Lemma \ref{graph}, there exist $k^*$ disjoint paths between $i$ and $j$. As the analysis in the proof of Theorem \ref{th cu22}, in each path, there exist at least one pair of adjacent agents $i_k$ and $i_{k+1}$, such that $||p_{i_k}-p_{i_{k+1}}||\rightarrow\infty$. Employing inequality (\ref{inequality cu22}), we have
\[
\begin{split}
V(p(0),q(0))&=||q(0)||^2+\frac{1}{2}\sum\limits_{i\in\mathcal{V}} \sum\limits_{j\in\mathcal{V}}\int_0^{||p_i(0)-p_j(0)||^2}G_{ij}\alpha(s)ds\\
&<k^*\int_0^{\infty}\alpha(s)ds\\
&\leq\frac{1}{2}\sum\limits_{i\in\mathcal{V}}\sum\limits_{j\in\mathcal{V}} \int_{0}^{||p_i-p_j||^2} G_{ij}\alpha(s)ds\\
&\leq V(p,q),
\end{split}
\]
a contradiction. Thus, $||p||$ is bounded for all $t\geq0$. We now proceed as in the proof of Theorem \ref{th cu22}.
\end{proof}

\begin{remark}
All the results above can be extended to general cases. More specifically, $\alpha_{ij}(\cdot)$ can be various for different pairs of agents. Each $\alpha_{ij}(s)$ is a continuous function of $s$ and is unnecessary to be nonincreasing. In this case, the condition for $\alpha(\cdot)$ in Theorem \ref{th cu22} is replaced by the condition that there exists a spanning tree with $\mathcal{E}'$ as its set of edges, and $\int_0^\infty\alpha_{ij}(s)ds=\infty$ for any $(i,j)\in\mathcal{E}'$. If this is not true, the initial states of all the agents are required to satisfy the following inequality,
\[
||v(0)||^2+ \frac{1}{2}\sum\limits_{i\in\mathcal{V}} \sum\limits_{j\in\mathcal{V}}\int_0^{||x_i(0)-x_j(0)||^2}a_{ij}\alpha(s)ds< d^*\min_{(i,j)\in\mathcal{E}'}\int_0^{\infty}\alpha_{ij}(s)ds,
\]
where $\mathcal{E}'$ is the set of edges associated with a spanning tree. The proof is similar to that of Theorem \ref{th cu22} and Corollary \ref{co cu22}, we omit it here.
\end{remark}

\subsection{Consensus with State-dependent Connectivity of Networks}

In this subsection, the connectivity of the communication graph is possibly broken due to the evolution of the agents. For realizing consensus, we always hope that the connectivity can be maintained. In the following ,we will use the Lyapunov method to search a specific condition for the initial states to guarantee the invariance of the connectivity. It is shown that under an intensive distribution of the agents' initial states, consensus can be finally reached.

Suppose Assumption 2 is satisfied. For agents with dynamics (\ref{ct1}), the following consensus protocol is considered,
\begin{equation}\label{ccu1}
u_i=\sum\limits_{j\in\mathcal{V}}\alpha_{ij} (x)(x_j-x_i).
\end{equation}
We present a sufficient condition for consensus by restricting the initial states of the agents. See the follows:
\begin{theorem}\label{th ccu1}
Consider a system consisting of $n$ agents with dynamics (\ref{ct1}). Under Assumption 2, suppose the following inequality holds.
\begin{equation}\label{ccu1con}
\frac{1}{2}\sum\limits_{i\in\mathcal{V}} \sum\limits_{j\in\mathcal{V}}\int_0^{||x_i(0)-x_j(0)||^2}\alpha(s)ds< (n-1)\int_0^{R^2}\alpha(s)ds.
\end{equation}
Then protocol (\ref{ccu1}) solves the average consensus asymptotically.
\end{theorem}
\begin{proof} Consider the Lyapunov function $V(x)=\frac{1}{2}||x||^2$, then $\dot{V}(x)=-x^TL_xx\leq0$. Thus $||x||\leq||x(0)||$. It follows the compactness of $\{x|~V(x)\leq V(x(0))\}$. By employing LaSalle's invariance principle, we have $L_xx\rightarrow0$ as $t\rightarrow\infty$. That is, for any different $i$ and $j$, $x_i=x_j$ or $||x_i-x_j||\geq R$ at the steady state.

Let $t\rightarrow\infty$, suppose consensus is not achieved, it is obvious that the communication graph $\mathcal{G}(t)$ is disconnected. By employing Lemma \ref{n-1}, there exist at least $n-1$ pairs of agents satisfying that the distance between any two agents in a pair is larger than or equal to $R$. Thus it holds that
\begin{equation}\label{steady}
\frac{1}{2}\sum\limits_{i\in\mathcal{V}} \sum\limits_{j\in\mathcal{V}}\int_0^{||x_i-x_j||^2}\alpha(s)ds\geq (n-1)\int_0^{R^2}\alpha(s)ds.
\end{equation}
We now consider the following function:
\begin{equation}\label{smooth W}
V_1(x)=\frac{1}{2}\sum\limits_{i\in\mathcal{V}} \sum\limits_{j\in\mathcal{V}}\int_0^{||x_i-x_j||^2}\alpha(s)ds.
\end{equation}
Differentiating $V_1(x)$, yields
\begin{equation}
\begin{split}
\dot{V}_1(x)&=\sum\limits_{i\in\mathcal{V}} \sum\limits_{j\in\mathcal{V}} \alpha_{ij}(x)(x_i-x_j)^T(u_i-u_j)\\
&=2x^TL_xu\\
&=-2x^TL_x^2x\leq0.
\end{split}
\end{equation}
Consequently, $V_1(x)\leq V_1(x(0))$ for all $t\geq0$. Together with (\ref{steady}), we have
$$
(n-1)\int_0^{R^2}\alpha(s)ds\leq\frac{1}{2}\sum\limits_{i\in\mathcal{V}} \sum\limits_{j\in\mathcal{V}}\int_0^{||x_i(0)-x_j(0)||^2}\alpha(s)ds,
$$
which is in contradiction with (\ref{ccu1con}). Therefore, consensus is achieved asymptotically. Since the communication graph is undirected, we have $\sum_{i\in\mathcal{V}}\dot{x}_i=0$ for $t\geq0$. Let $x^*$ be the steady state of each agent, then $nx^*=\sum_{i\in\mathcal{V}}x_i(0)$. Therefore, the consensus state is the average of the initial states.
\end{proof}

For agents with dynamics (\ref{ct2}), the following static consensus control law is considered.
\begin{equation}\label{ccu2}
u_i=-kv_i+ \sum\limits_{j\in\mathcal{V}} \alpha_{ij}(x)(x_j-x_i).
\end{equation}
\begin{theorem}\label{th ccu2}
Consider a system consisting of $n$ agents with dynamics (\ref{ct2}). Under Assumption 2, suppose the following inequality holds:
\begin{equation}\label{ccu2con}
||v(0)||^2+\frac{1}{2}\sum\limits_{i\in\mathcal{V}} \sum\limits_{j\in\mathcal{V}} \int_0^{||x_i(0)-x_j(0)||^2} \alpha(s)ds< (n-1)\int_0^{R^2}\alpha(s)ds.
\end{equation}
Then protocol (\ref{ccu2}) solves the consensus problem asymptotically.
\end{theorem}
\begin{proof} Consider the Lyapunov function (\ref{lyapunov cu22}) by replacing $p$ and $q$ with $x$ and $v$. As the same way in the proof of Theorem \ref{th cu22}, we obtain that $\dot{V}(x,v)=-k\sum\limits_{i\in\mathcal{V}}||v_i||^2\leq0$. By employing the function (\ref{lyapunov cu21}), we know that $x$ and $v$ are both bounded. Then it follows the compactness of $\{x,v|V(x,v)\leq V(x(0),v(0))\}$.
From LaSalle's invariance principle, if $t\rightarrow\infty$, one has $v_i\rightarrow0$ for any $i\in\mathcal{V}$. That is, $\dot{v}_i\rightarrow0$, implying that $||x_i-x_j||=0$ or $||x_i-x_j||\geq R$ at the steady state. Suppose consensus is not achieved in the steady state. From Lemma \ref{n-1}, for $t\rightarrow\infty$, the following holds.
$$
||v||^2+\frac{1}{2}\sum\limits_{i\in\mathcal{V}} \sum\limits_{j\in\mathcal{V}}\int_0^{||x_i-x_j||^2}\alpha(s)ds\geq (n-1)\int_0^{R^2}\alpha(s)ds.
$$
Together with $V(x,v)\leq V(x(0),v(0))$, we have
$$
||v(0)||^2+\frac{1}{2}\sum\limits_{i\in\mathcal{V}} \sum\limits_{j\in\mathcal{V}}\int_0^{||x_i(0)-x_j(0)||^2}\alpha(s)ds\geq (n-1)\int_0^{R^2}\alpha(s)ds,
$$
a contradiction with (\ref{ccu1con}). Therefore, all the agents achieve consensus asymptotically.
\end{proof}

\section{Consensus of Discrete-time Multi-agent Systems}

In this section, the consensus problem of discrete-time multi-agent systems with state-dependent information transmission laws is considered. Different from the case of continuous-time, the discontinuity of the control input can be adopted.

\subsection{Discontinuous State-dependent Transmission Weight}

Similar to the one of continuous-time systems, we use a function $\alpha(\cdot)$ to interpret the relationship between the transmission weight and the relative difference between agents' states. The previous assumptions are modified as follows by relaxing the continuity of $\alpha(\cdot)$.

\emph{Assumption 3}: $\alpha(\cdot): \mathbb{R}_{\geq0}\rightarrow \mathbb{R}_{>0}$ is nonincreasing, $\alpha(0)<\infty$.

\emph{Assumption 4}: $\alpha(\cdot): \mathbb{R}_{\geq0}\rightarrow \mathbb{R}_{\geq0}$ is nonincreasing, $\alpha(0)<\infty$, $\alpha(s)>0$ if $s<R^2$, $\alpha(s)=0$ if $s\geq R^2$, where $R\in\mathbb{R}_{>0}$ is a constant.

\subsection{A Lyapunov-like Function}

Before entering into our results, we introduce a function $w(z): \mathbb{R}_{\geq0}\rightarrow \mathbb{R}_{\geq0}$ which will be used to construct the Lyapunov function.
\[
w(z)=\left\{\begin{array}{l}
\alpha(r)z, ~~~~~~~~~~~~~~~~~~~~~~~~~~~~~~~ 0\leq z< r,\\
\sum\limits_{s=1}^{\lfloor\frac{z}{r}\rfloor}\alpha(sr)r+ \alpha(\lceil\frac{z}{r}\rceil r)(z-\lfloor\frac{z}{r}\rfloor r), ~~~~~~~~z\geq r,
\end{array}\right.
\]
where $\alpha(z)$ is nonincreasing of $z$, $r$ is a positive constant. For better understanding $w(z)$, we present an example with $r=1$ to express the relationship between $w(\cdot)$ and $\alpha(\cdot)$. The area of the shaded part of Fig. \ref{w(z)1} is equal to $w(3.5)$, while the area of the shaded part of Fig. \ref{w(z)2} is equal to $w(0.5)$. For simplicity, we define $x_{ij}(t)=x_i(t)-x_j(t)$, $W_{ij}(t)=w(||x_{ij}(t)||^2)$, $W(t)=\frac{1}{2}\sum\limits_{i\in\mathcal{V}} \sum\limits_{j\in\mathcal{V}}G_{ij}W_{ij}(t)$.
\begin{figure}
\begin{minipage}[t]{0.5\linewidth}
\centering
\includegraphics[width=0.6\textwidth]{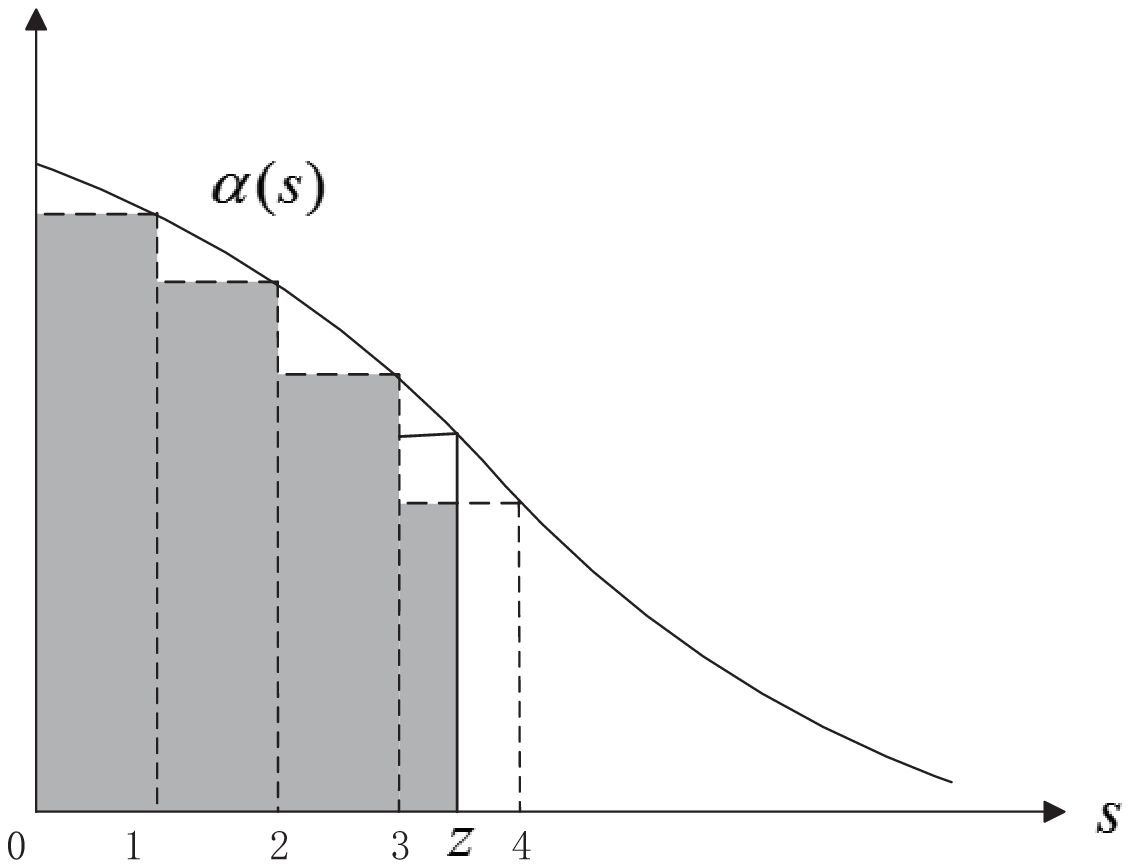}
\caption{$w(z)$ with $r=1,~z=3.5$.}\label{w(z)1}
\end{minipage}
\begin{minipage}[t]{0.5\linewidth}
\centering
\includegraphics[width=0.6\textwidth]{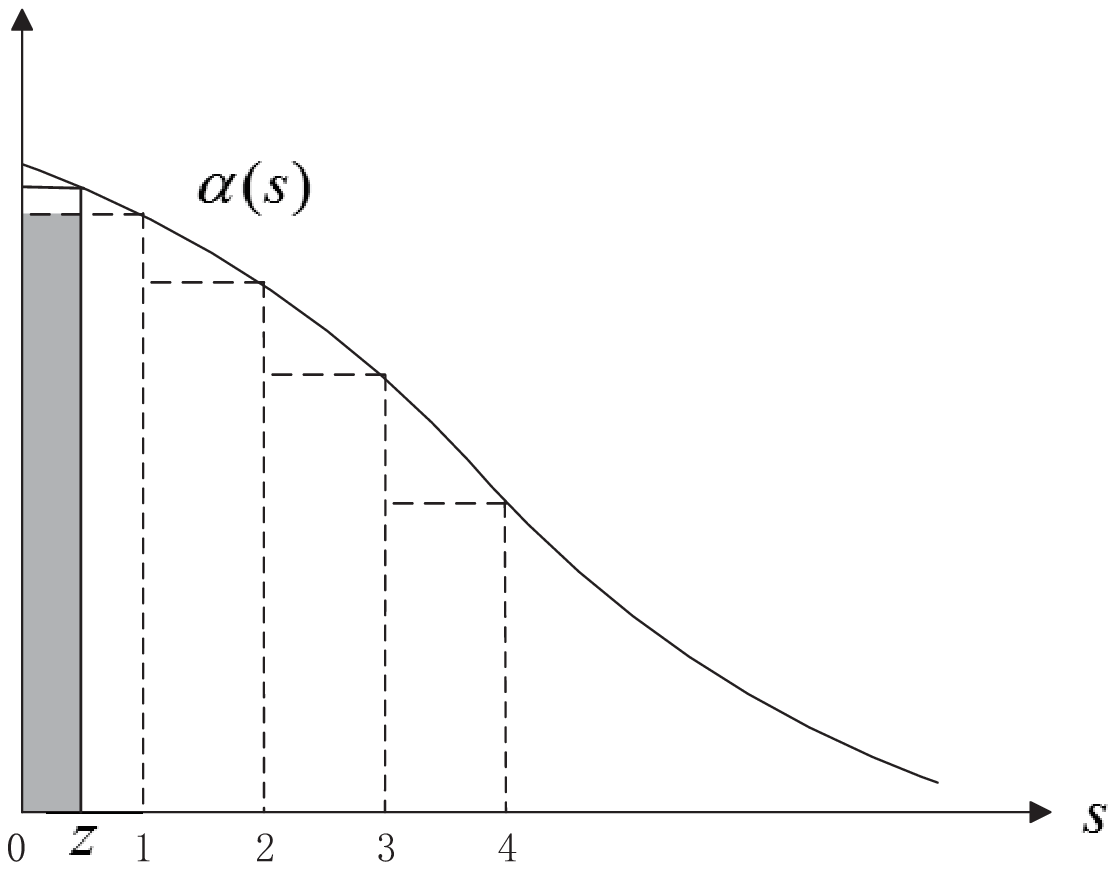}
\caption{$w(z)$ with $r=1,~z=0.5$.}\label{w(z)2}
\end{minipage}
\end{figure}

The following Proposition shows some properties of $W$, which will be important for the main results. The corresponding proof is shown in Section 7.
\begin{proposition}\label{prW}
For any $z\geq0$, the following hold.

(1). Suppose that the communication graph $\mathcal{G}$ is connected. Then $W(t)\geq0$, $W(t)=0$ if and only if $x_i=x_j$ for any $i,j\in\mathcal{V}$.

(2). For a fixed $r$, $w(z)$ is increasing of $z$.

(3). For all $t\geq0$,
\begin{equation}\label{dt1W}
W(t+1)-W(t)\leq \frac{1}{2}\sum\limits_{i\in\mathcal{V}} \sum\limits_{j\in\mathcal{V}} G_{ij} \alpha_{ij}(x(t))(||x_{ij}(t+1)||^2-||x_{ij}(t)||^2).
\end{equation}

(4). $\lim\limits_{r\rightarrow0}w(z) =\int_0^z\alpha(s)ds$ for $0\leq z<\infty$.
\end{proposition}
\begin{remark}\label{re w(z)}
We can see that $w(z)$ is the approximation of $\int_0^z\alpha(s)ds$ in some sense. And that with the decreasing of $r$, $w(z)$ is more closer to $\int_0^z\alpha(s)ds$. Actually, when we let $w(z)=\int_0^z\alpha(s)ds$, $(1)$, $(2)$ and $(3)$ in Proposition \ref{pr symmetry} also hold. The corresponding proof is similar. In the rest of this paper, we admit $w(z)=\int_0^z\alpha(s)ds$ for $r=0$.
\end{remark}

\subsection{Consensus with Fixed Connectivity of Networks}

For agents with dynamics (\ref{dt1}), the consensus protocol is given by
\begin{equation}\label{du1}
u_i(t)=h\sum\limits_{j=1}^nG_{ij} \alpha_{ij}(x(t))(x_j(t)-x_i(t)),
\end{equation}
where $h>0$ is the control gain.
\begin{theorem}\label{th du1}
Consider a system consisting of $n$ agents with dynamics (\ref{dt1}). Under Assumption 3, protocol (\ref{du1}) globally asymptotically solves the average consensus problem if the communication topology is connected and $h<\frac{1}{d_{max}\alpha(0)}$, where $d_{max}$ is the maximum degree of all the agents.
\end{theorem}
\begin{proof} Consider $V(t)=||p(t)||^2$ as a Lyapunov function. It is obvious that $V$ is positive definite. And
\[
\begin{split}
V(t+1)-V(t)&=p^T(t)(I-hL_t)^2p(t)-p^T(t)p(t)\\
&=p^T(-2hL_t+h^2L_t^2)p.
\end{split}
\]
Let $\Xi_t=-2hL_t+h^2L_t^2$, the eigenvalues of $\Xi_t$ are denoted by $\xi_i=-2h\lambda_i(L_t)+h^2\lambda_i^2(L_t)=h\lambda_i(L_t)(h\lambda_i(L_t)-2)$, and it is straightforward to see that the eigenspace of $\Xi_t$ corresponding to $\xi_i$ is similar to the one of $L_t$ corresponding to $\lambda_i(L_t)$ for any $i\in\mathcal{V}$. From Gerschgorin Theorem, $\lambda_i(L_t)\leq \max\limits_{i\in\mathcal{V}} \{2\sum\limits_{j\in\mathcal{N}_i} \alpha_{ij}(x)\} \leq 2d_{max} \alpha(0)$. Therefore, $h\lambda_i(L_t)-2<\frac{1}{d_{max}\alpha(0)}\cdot 2d_{max}\alpha(0)-2=0$. Thus, $V(t+1)-V(t)=p^T\Xi_tp\leq0$. That is, $\Xi_t$ is negative definite. Since graph $\mathcal{G}$ is connected, together with Lemma \ref{rankL} and Lemma \ref{eigenvalue}, one has $H_0(\Xi_t\otimes I_m)=H_0(L_t\otimes I_m)=M$, and
\[
\begin{split}
V(t+1)-V(t)&\leq-\lambda_2(-\Xi_t) ||p-\pi_{H_0(\Xi_t\otimes I_m)}(p)||^2\\
&=-\lambda_2(-\Xi_t)||p-\pi_M(p)||^2 \\
&= -\lambda_2(-\Xi_t)||p||^2\leq0,
\end{split}
\]
where $\lambda_2(-\Xi_t)$ is the smallest nonzero eigenvalue of $-\Xi_t$. Then $||p||$ is bounded by $||p(0)||$. From Lemma \ref{bound}, $\lambda_2(L_t)$ is lower bounded, which implies that all the nonzero eigenvalues have lower bounded. Together with $\lambda_i(L_t)\leq 2d_{max} \alpha(0)$, there exists a constant $c>0$ such that $|\lambda_i(L_t)|\leq c$. Hence, there exists a $c'<0$ such that $-\lambda_2(-\Xi_t)\leq c'$. Then $V(t+1)-V(t)$ is negative definite. From Lyapunov's Theorem, $p\rightarrow0$ as $t\rightarrow\infty$. Note that $\sum_{i\in\mathcal{V}}x_i(t+1)=\sum_{i\in\mathcal{V}}x_i(t)$ in every step, which results in $\lim\limits_{t\rightarrow\infty}x_i(t)= \frac{1}{n}\sum_{i\in\mathcal{V}}x_i(0)$. That is, the average consensus is achieved. Together with the radial unboundedness of $V$, the conclusion is global.
\end{proof}

For agents with dynamics (\ref{dt2}), the following protocol is considered,
\begin{equation}\label{du2}
u_i(t)=-k_2 v_i(t)+ k_3 \sum\limits_{j=1}^nG_{ij} \alpha_{ij}(x(t))(x_j(t)-x_i(t)),
\end{equation}
where $k_2$, $k_3 >0$, $i=1,\cdots,n$.
\begin{theorem}\label{th du2}
Consider a system consisting of $n$ agents with dynamics (\ref{dt2}). Under Assumption 3, protocol (\ref{du2}) globally asymptotically solves the consensus problem if the communication graph is connected, and the following conditions for $k_1, k_2 , k_3 $ are satisfied,
\begin{equation}\label{d2c1}
k_2 <min\{2, k_1+1\},
\end{equation}
\begin{equation}\label{d2c2}
k_3 <min\{\frac{k_2 (2-k_2 )}{2d_{max}\alpha(0)k_1(k_1-k_2 +1)}, \frac{k_2 }{d_{max}\alpha(0)(k_1+1)}\}.
\end{equation}
Specifically, if the sum of the initial velocity of each agent is zero, the average consensus problem is solved.
\end{theorem}
\begin{proof} Consider the following function as a Lyapunov function candidate,
\begin{equation}\label{lyapunov du21}
V(t)=||k_2 x+k_1v||^2+k_1||v||^2+ \frac{1}{2}k_3(k_1+1-k_2 )\sum\limits_{i\in\mathcal{V}} \sum\limits_{j\in\mathcal{V}} G_{ij}W_{ij}(t).
\end{equation}
Employing Proposition \ref{prW}, we have
\begin{equation}\label{t+1-t}
\begin{split}
V(t+1)-V(t)&\leq ||k_2 x(t+1)+k_1v(t+1)||^2- ||k_2 x(t)+k_1v(t)||^2\\
&~~~~+k_1||v(t+1)||^2-k_1||v(t)||^2\\
&~~~~+k_1^2k_3(k_1+1-k_2)v^TL_tv \\
&~~~~+2k_1k_3(k_1+1-k_2)x^TL_tv \\
&=v^T[k_1k_2(k_2-2)I+k_1^2k_3(k_1+1-k_2)L_t]v\\
&~~~~+x^T[-2k_1k_2k_3 L_t+(k_1^2+h)k_3 ^2L_t^2]x.
\end{split}
\end{equation}
Let $\Xi_{1t}=k_1k_2(k_2-2)I+k_1^2k_3(k_1+1-k_2)L_t$ and $\Xi_{2t}=-2k_1k_2k_3 L_t+(k_1^2+h)k_3 ^2L_t^2$. Then $V(t+1)-V(t)\leq0$ if $v^T\Xi_{1t}v+x^T\Xi_{2t}x\leq0$. To achieve this, we just require the following inequalities for any $i\in\mathcal{V}$.
\begin{equation}\label{d2cc1}
k_1k_2 (k_2 -2)+k_1^2k_3 (k_1+1-k_2 )\lambda_i<0,
\end{equation}
\begin{equation}\label{d2cc2}
-2k_1k_2 k_3 +(k_1^2+h)k_3 ^2\lambda_i<0.
\end{equation}
By Gerschgorin Theorem, it holds that $\lambda_i\leq \max\limits_{i\in\mathcal{V}} \{2\sum\limits_{j\in\mathcal{N}_i} \alpha_{ij}(x)\} \leq 2d_{max} \alpha(0)$. Hence, conditions (\ref{d2c1}) and (\ref{d2c2}) lead to (\ref{d2cc1}) and (\ref{d2cc2}). Consequently, $V(t+1)-V(t)\leq0$.

Since $k_2 <k_1+1$, together with the nonnegativity of $W_{ij}$ and the definition of $V$ in (\ref{lyapunov du21}), one has $\sqrt{V}\geq||k_2 x+k_1v||\geq||k_2 x||-||k_1v||$, and $\sqrt{V}\geq\sqrt{k_1}||v||$. Then $\sqrt{V}+2\sqrt{k_1V}\geq||k_2 x||+k_1||v||$. Therefore, if $||x||^2+||v||^2\rightarrow\infty$, we have
$$
V\geq\frac{\min\{k_2,k_1\}(||x||^2+||v||^2)} {(1+2\sqrt{k_1})^2}\rightarrow\infty.
$$
It follows the radial unboundedness of $V$ and the compactness of the invariant set $\Omega=\{V(t)\leq V(0)\}$. Invoking LaSalles's invariance principle, $v^T\Xi_{1t}v+x^T\Xi_{2t}x\rightarrow0$ as $t\rightarrow\infty$. Note that $v^T\Xi_{1t}v=0$ if and only if $v=0$, while $x^T\Xi_{2t}x=0$ if and only if $x\in H_0(\Xi_{2t})$. From the connectivity of graph $\mathcal{G}$ and Lemma \ref{rankL}, $H_0(\Xi_{2t})=H_0(L_t)=M$. Consequently, the position states of the agents globally asymptotically achieve consensus, and the velocity states of the agents globally asymptotically converge to the origin.

Now we explore the consensus state for the group of agents. Consider $U(t)=\frac{1}{k_3 }\sum_{i\in\mathcal{V}}v_i(t)+ \frac{k_2 }{k_1k_3 }\sum_{i\in\mathcal{V}}x_i(t)$. One has
$$
U(t+1)-U(t)=-\frac{k_2 }{k_3 }\sum_{i\in\mathcal{V}} v_i +\sum_{i\in\mathcal{V}} \sum_{j\in\mathcal{V}}G_{ij}\alpha_{ij}(x_j-x_i) +\frac{k_2}{k_3}\sum_{i\in\mathcal{V}} v_i=0.
$$
Therefore, let $x^*$ denote the consensus state, it follows that $U(x(0),v(0))=U(x^*,0)$. We finally have $x^*=\displaystyle{\frac{k_1k_3} {nk_2}U(0)}$.

Moreover, if $\sum_{i\in\mathcal{V}}v_i(0)=0$, it is clear that $x^*=\frac{1}{n}\sum_{i\in\mathcal{V}}x_i(0)$.
\end{proof}
\begin{remark}\label{re dt1}
$\alpha_{ij}(\cdot)$ in this subsection can be multiple for different pairs of agents and each one satisfies Assumption 3. If this change happens, let $\alpha_{max}(0)=\max\limits_{i,j\in\mathcal{V}}\alpha(0)$, the condition of $h$ in Theorem \ref{th du1} becomes to be $h<\frac{1}{d_{max}\alpha_{max}(0)}$ instead. The rest of the conclusions are undisturbed and the corresponding proofs are the same. For Theorem \ref{th du2}, there are various $w_{ij}$ due to different $\alpha_{ij}(\cdot)$, then $W(t)=\frac{1}{2}\sum\limits_{i\in\mathcal{V}} \sum\limits_{j\in\mathcal{V}}G_{ij}w_{ij}(||x_{ij}(t)||^2)$. By the similar approach, we can obtain the same result as Theorem \ref{th du2} except for replacing $\alpha(0)$ in (\ref{d2c2}) with $\alpha_{max}(0)$.
\end{remark}

\subsection{Consensus with State-dependent Connectivity of Networks}

For agents with dynamics (\ref{dt1}), the consensus protocol is given by
\begin{equation}\label{ddu1}
u_i(t)=h\sum\limits_{j=1}^n \alpha_{ij}(x(t))(x_j(t)-x_i(t)),
\end{equation}
where $h>0$ is the control gain.
\begin{theorem}\label{th ddu1}
Consider a system consisting of $n$ agents with dynamics (\ref{dt1}). Under Assumption 4, suppose $h<\frac{1}{(n-1)\alpha(0)}$, and there exists an $r\in[0,R^2)$, such that
\begin{equation}\label{ddu1con}
W(0)<(n-1)w(R^2).
\end{equation}
Then protocol (\ref{ddu1}) asymptotically solves the average consensus problem.
\end{theorem}
\begin{proof} Suppose (\ref{ddu1con}) holds. Consider $V(t)=||x(t)||^2$ as a Lyapunov function candidate, one has $V(t+1)-V(t)=x^T(-2hL_t+h^2L_t^2)x$. Let $\Xi_t=-2hL_t+h^2L_t^2$, the eigenvalues of $\Xi_t$ are $\xi_i=h\lambda_i(L_t)(h\lambda_i(L_t)-2)\leq h\lambda_i (\frac{1}{n-1}\alpha(0)\cdot 2d_{max}\alpha(0)-2)\leq0$, $i\in\mathcal{V}$. Then $\Omega=\{x~|~||x||\leq||x(0)||\}$ is positively invariant and compact. Consequently, $V(t+1)-V(t)\rightarrow0$ as $t\rightarrow\infty$. That is, $x_i=x_j$ or $||x_i-x_j||\geq R$ when $t\rightarrow\infty$. Suppose consensus is not reached. Employing Lemma \ref{n-1}, we have $W(t)\geq(n-1)w(R^2)$ as $t\rightarrow\infty$.

For $W(t)$, from Proposition \ref{prW}, the following holds,
\[
\begin{split}
W(t+1)-W(t)&\leq u^TL_tu+2x^TL_tu\\
&= h^2x^TL_t^3x-2hx^TL_t^2x.
\end{split}
\]
Since $h^2\lambda_i^3(L_t)-2h\lambda_i^2(L_t)= \lambda_i(L_t)\xi_i\leq0$, we have $\frac{1}{2}(n-1)w(R^2)\leq W(t)\leq W(0)$, which conflicts with (\ref{ddu1con}). We then obtain the conclusion.
\end{proof}

For agents with dynamics (\ref{dt2}), the following protocol is considered:
\begin{equation}\label{ddu2}
u_i(t)=-k_2 v_i(t)+ k_3 \sum\limits_{j=1}^n \alpha_{ij}(x(t))(x_j(t)-x_i(t)),
\end{equation}
where $k_2 $, $k_3 >0$, $i=1,\cdots,n$.
\begin{theorem}\label{th ddu2}
Consider a system consisting of $n$ agents with dynamics (\ref{dt2}). Under Assumption 4, suppose that $k_1, k_2$ and $k_3 $ satisfy (\ref{d2c1}) and
\begin{equation}\label{d21c2}
k_3 <min\{\frac{k_2 (2-k_2 )}{2(n-1)\alpha(0)k_1(k_1-k_2 +1)}, \frac{k_2 }{(n-1)\alpha(0)(k_1+1)}\}.
\end{equation}
And there exists an $r\in[0,R^2)$, such that
\begin{multline}\label{ddu2con}
||k_2 x(0)||^2+2k_1k_2 x(0)^Tv(0)+(k_1^2+k_1)||v(0)||^2+ \\ \frac{1}{2}\sum\limits_{i\in\mathcal{V}} \sum\limits_{j\in\mathcal{V}}(k_1+1-k_2 )k_3 W_{ij}(0) < (k_1+1-k_2)k_3 (n-1) w(R^2).
\end{multline}
Then protocol (\ref{ddu2}) asymptotically solves the consensus problem.
\end{theorem}
\begin{proof} Suppose condition (\ref{ddu2con}) holds. Let $G_{ij}=1$ for any $i,j\in\mathcal{V}$, (\ref{lyapunov du21}) is considered as the Lyapunov function candidate. From the radial unboundedness of $V$, we get the compactness of $\{x,v|~V(t)\leq V(0)\}$. Due to the fact that (\ref{d2c1}) and (\ref{d2c2}) are satisfied, together with (\ref{t+1-t}), one has $V(t+1)-V(t)\leq0$, and $V(t+1)-V(t)=0$ if and only if $v^T\Xi_{1t}v+x^T\Xi_{2t}x=0$. By invoking LaSalle's invariance principle, we have $v^T\Xi_{1t}v+x^T\Xi_{2t}x\rightarrow0$ when $t\rightarrow\infty$. That is, $v_i\rightarrow0$, $||x_i-x_j||\rightarrow0$ or $||x_i-x_j||\geq R$, for $t\rightarrow\infty$. Suppose consensus is not achieved. By employing Lemma \ref{n-1}, it follows that
\[ \begin{split}
V(0)\geq V(t)&\geq\frac{1}{2}\sum\limits_{i\in\mathcal{V}} \sum\limits_{j\in\mathcal{V}}(k_1+1-k_2 )k_3 W_{ij}(t) \\&\geq \frac{1}{2} (k_1+1-k_2)k_3 (n-1)w(R^2),
\end{split}
\]
as $t\rightarrow\infty$. This contradicts with (\ref{ddu2con}). Therefore consensus is achieved asymptotically.
\end{proof}

\begin{remark}\label{contiuity}
Under Assumption 4, note that $W(R^2)$=0 if $r=R^2$. Then (\ref{ddu1con}) and (\ref{ddu2con}) will never be satisfied. Therefore, $r<R^2$ is necessary in Theorem \ref{th ddu1} and Theorem \ref{th ddu2}. Moreover, when $r$ is changed, the validy of (\ref{ddu1con}) or (\ref{ddu2con}) may also be changed. Although smaller $r$ make $W(R^2)$ larger, but it does not mean that smaller $r$ is more possible to satisfy the conditions, because $W(z)$ will also become larger. The examples in Section 5.2 will show us this in detail.
\end{remark}

\section{Applications and Simulations}

\subsection{Applications to the Transmission Law of C-S Model}

In C-S model \cite{Cucker07}, the communication weight between any two agents is set as
\begin{equation}\label{csaij}
a_{ij}=\frac{H}{(1+||x_i-x_j||^2)^\beta},
\end{equation}
where $H>0$ and $\beta\geq0$ are system parameters. That is, $\alpha(s)=\frac{H}{(1+s)^\beta}$, $\mathcal{G}$ is a complete graph.

We now solve the consensus problem for a group of mobile agents applying (\ref{csaij}) as the information transmission weight.

For agents with single integrator dynamics and protocol (\ref{cu1}), Fig. \ref{fig cu1} describes the evolution of the agents, which consists of $30$ agents with random initial states. Fig. \ref{fig du1} gives the simulation of the system (\ref{dt1}) with protocol (\ref{du1}).

\begin{figure}
\begin{minipage}[t]{0.5\linewidth}
\centering
\includegraphics[width=0.8\textwidth]{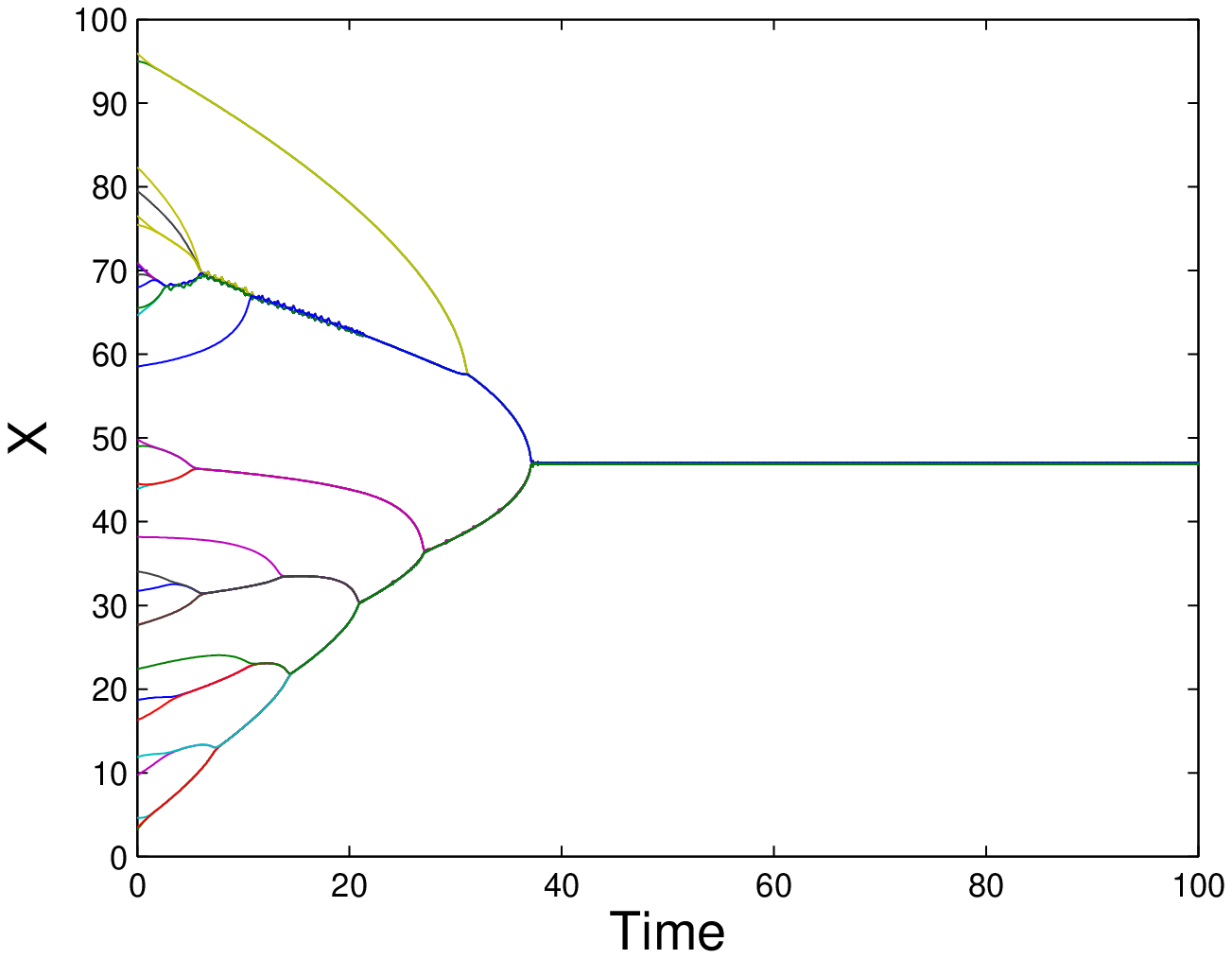}
\caption{Agents with dynamics (\ref{ct1}) and protocol (\ref{cu1}).}\label{fig cu1}
\end{minipage}
\begin{minipage}[t]{0.5\linewidth}
\centering
\includegraphics[width=0.8\textwidth]{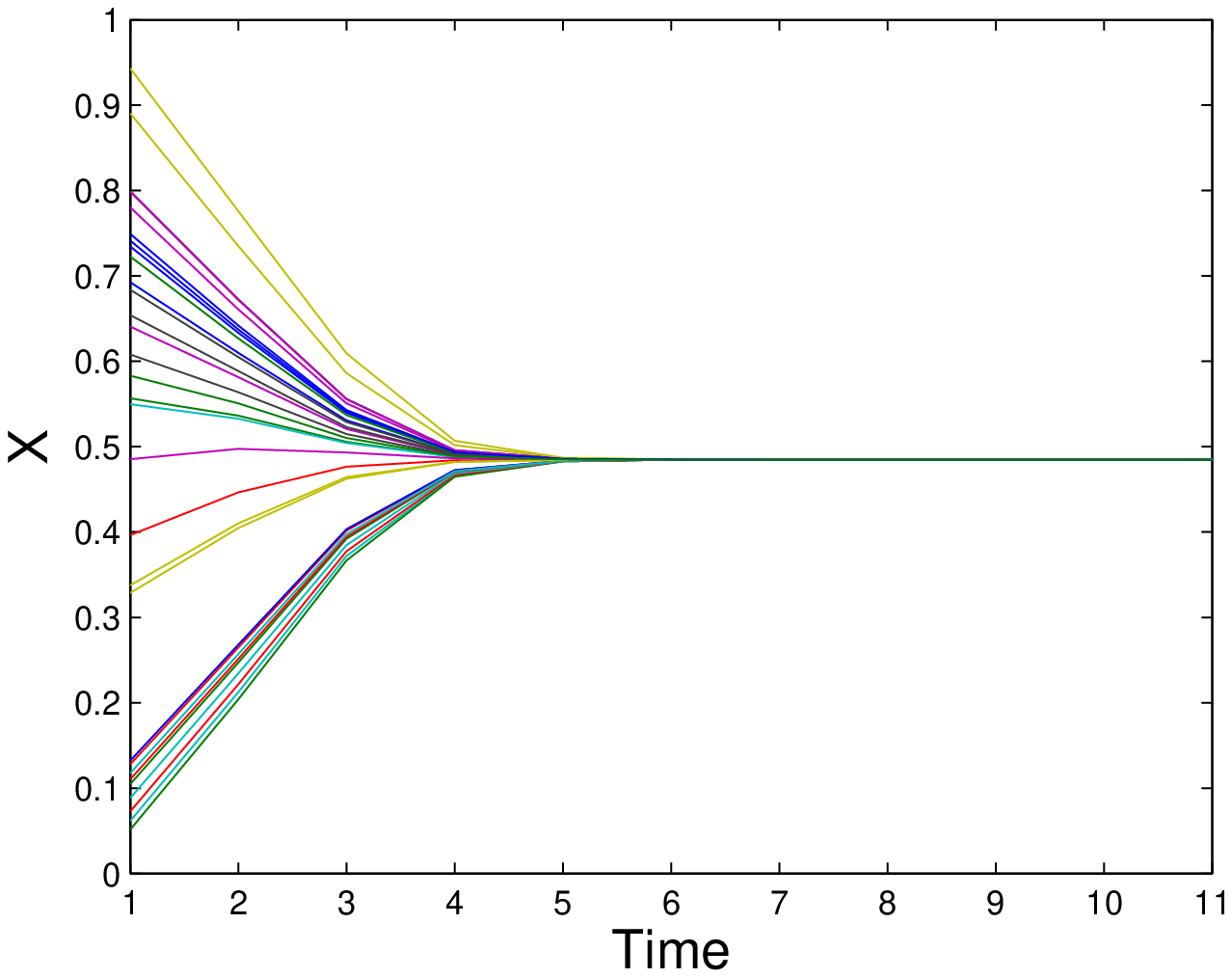}
\caption{Agents with dynamics (\ref{dt1}) and protocol (\ref{du1}).}\label{fig du1}
\end{minipage}
\end{figure}

For agents with double integrator dynamics (\ref{ct2}), we consider a multi-agent system consisting of $6$ agents, each agent is of dynamics (\ref{ct2}) and employs protocol (\ref{cu21}) with $\alpha(s)=\frac{H}{(1+s)^{\beta}}$, $H=1$, $\beta=3$, $k=1$, $\mathcal{G}$ is a complete graph. According to Theorem \ref{th cu21}, consensus can be achieved under arbitrary initial states. Fig. \ref{fig cu21} shows the results. Moreover, by employing the same $\alpha(s)$ with $H=1$ and $\beta=1$, let $k_1=1$, $k_2=1.5$, $k_3=0.14$. Fig. \ref{fig du2} describes the evolution of the agents with dynamics (\ref{dt2}).
\begin{figure}[htbp]
\centering
\includegraphics[width=0.4\textwidth]{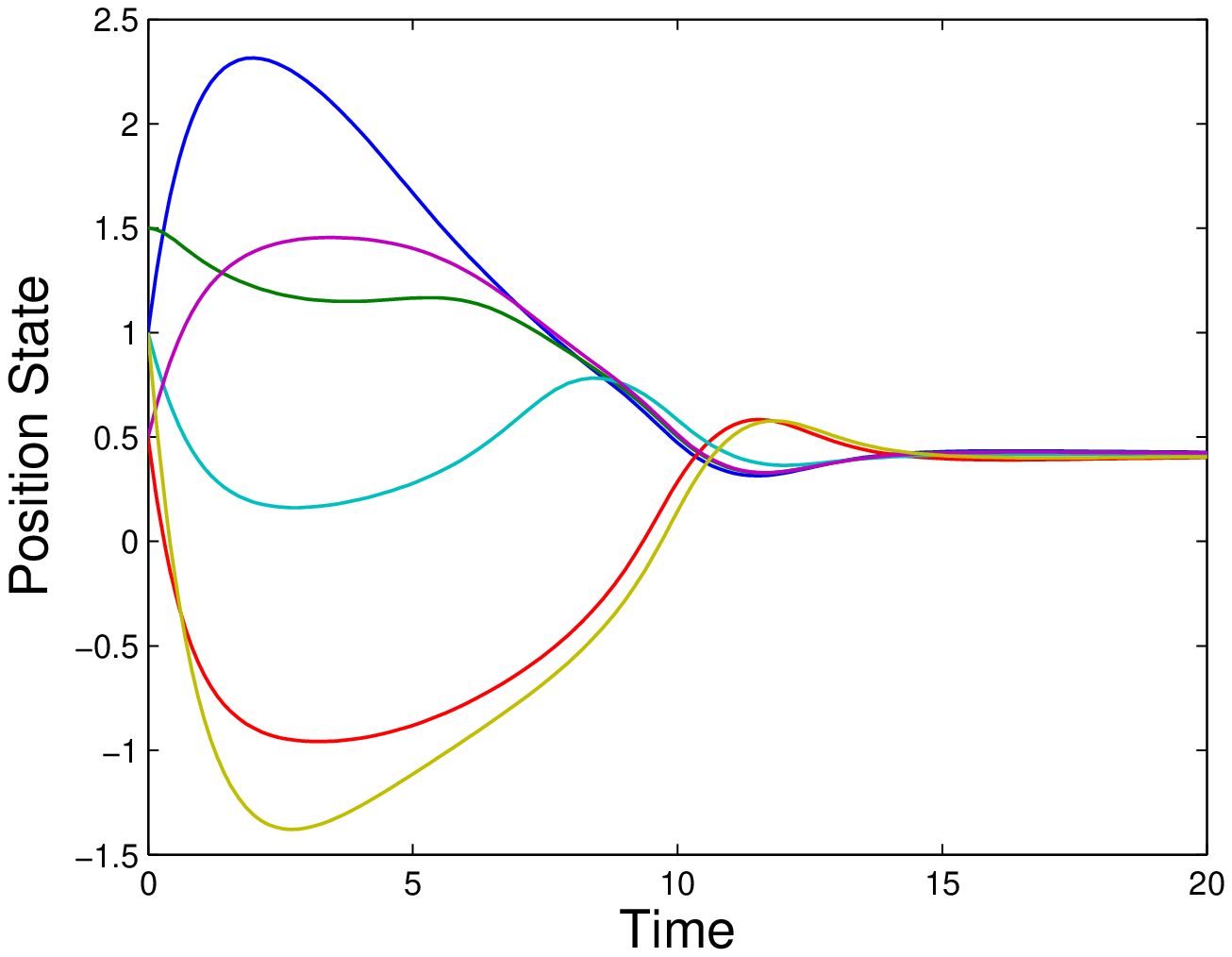}
\includegraphics[width=0.4\textwidth]{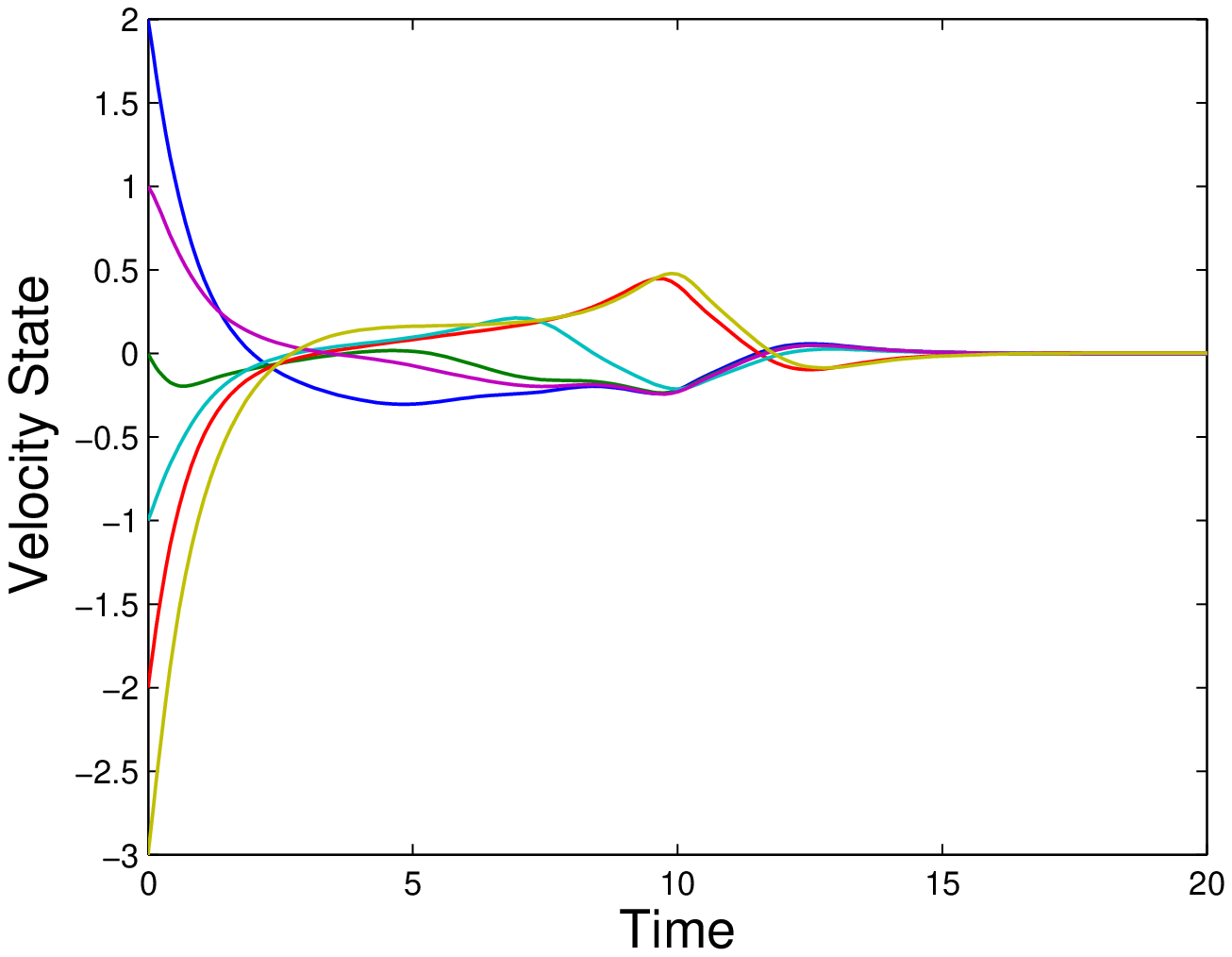}
\caption{Agents with dynamics (\ref{ct2}) and protocol (\ref{cu21}).}\label{fig cu21}
\end{figure}
\begin{figure}[htbp]
\centering
\includegraphics[width=0.4\textwidth]{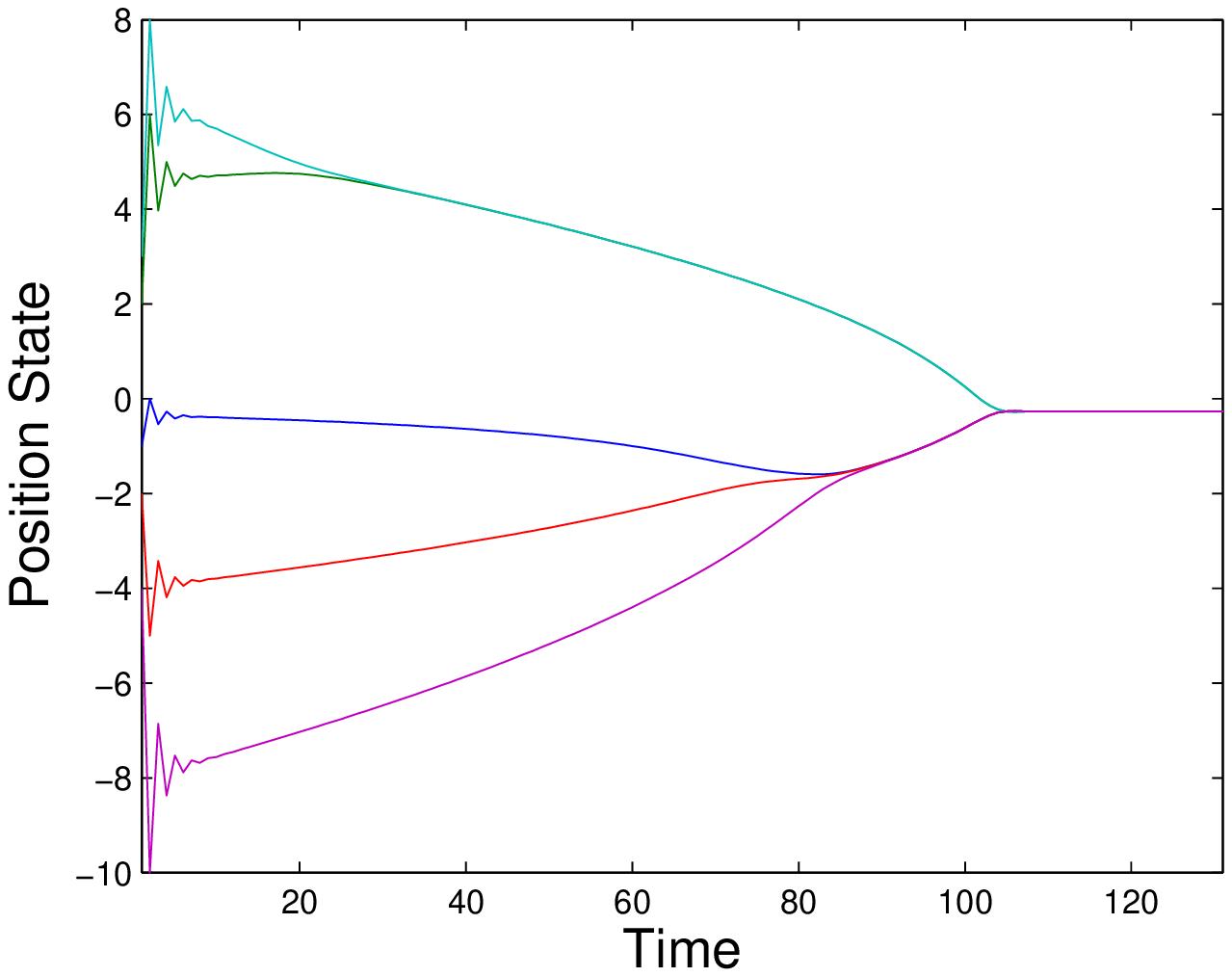}
\includegraphics[width=0.4\textwidth]{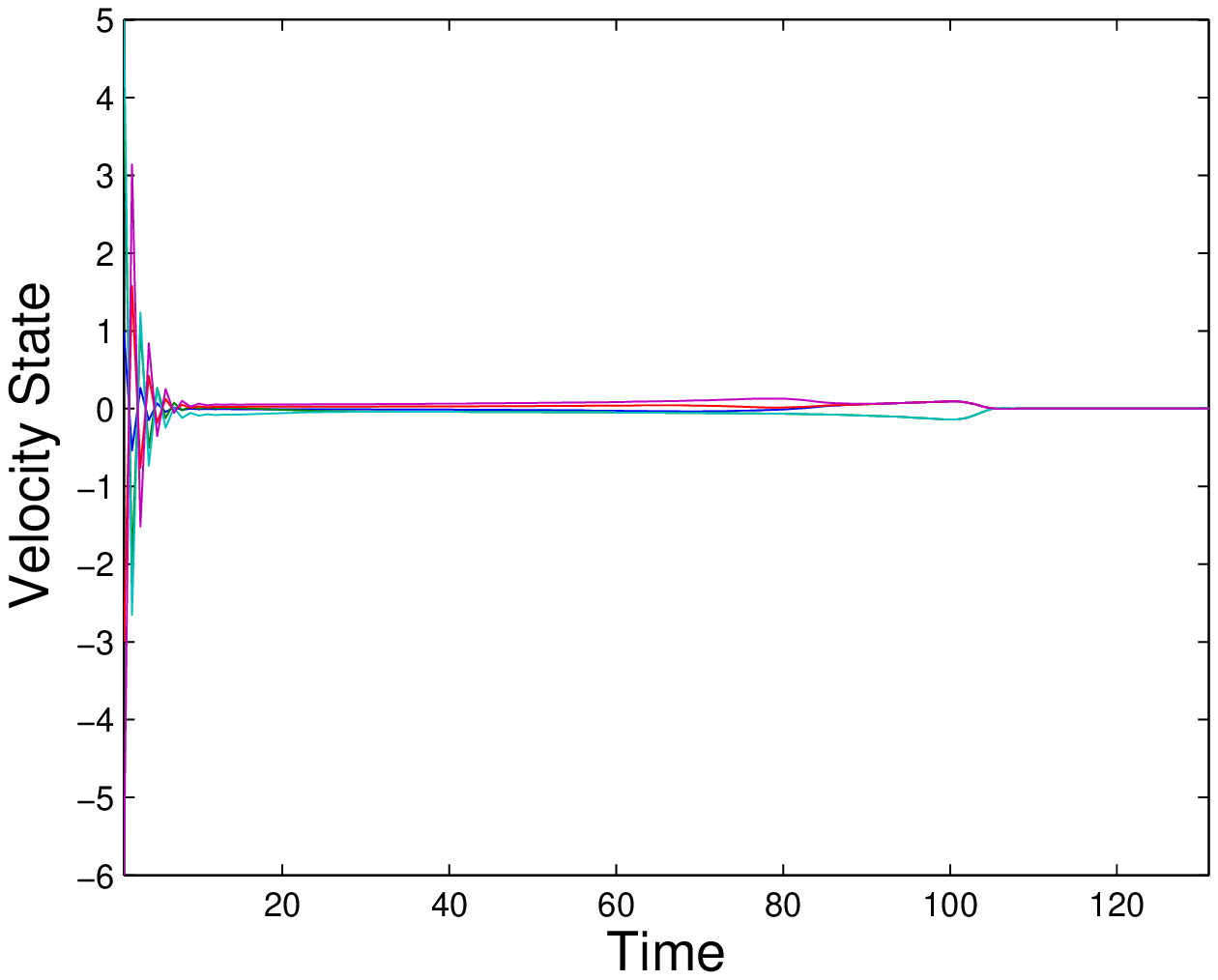}
\caption{Agents with dynamics (\ref{dt2}) and protocol (\ref{du2}).}\label{fig du2}
\end{figure}

When protocol (\ref{cu22}) is applied, it is necessary to explore a condition for $\alpha$ to solve the consensus problem. According to Theorem \ref{th cu22}, one just requires $\int_0^\infty\alpha(s)ds=\infty$ to realize consensus. Note that
\[
\int_0^\infty\frac{H}{(1+s)^\beta}ds=\begin{cases}
\frac{H}{1-\beta}(1+s)^{1-\beta}\big|_0^\infty, & \beta\neq1,\\
H\ln(1+s)\big|_0^\infty, & \beta=1.
\end{cases}
\]
Therefore, if $\beta\leq1$, then $\int_0^\infty\alpha(s)ds=\infty$. That is, the average consensus is asymptotically reached. Otherwise, if $\beta>1$, $\int_0^\infty\alpha(s)ds<\infty$, due to Corollary \ref{co cu22}, the average consensus is achieved if the following inequality holds:
\begin{equation}\label{cscon}
||v(0)||^2+ \frac{1}{2}\sum\limits_{i\in\mathcal{V}} \sum\limits_{j\in\mathcal{V}}\int_0^{||x_i(0)-x_j(0)||^2}\alpha(s)ds< (n-1)\int_0^{\infty}\alpha(s)ds.
\end{equation}
Now we investigate a system consisting of $6$ agents with dynamics (\ref{ct2}) and protocol (\ref{cu22}), the initial states of the agents and $\alpha(s)$ are chosen the same as the ones in the last example. It is clear that consensus is failed to be reached in Fig. \ref{fig cu22f} since condition (\ref{cscon}) is not satisfied. When we set $H=150$ and $\beta=3$, (\ref{cscon}) is guaranteed and the average consensus is asymptotically achieved, as shown in Fig. \ref{fig cu22s}.

\begin{figure}[htbp]
\centering
\includegraphics[width=0.4\textwidth]{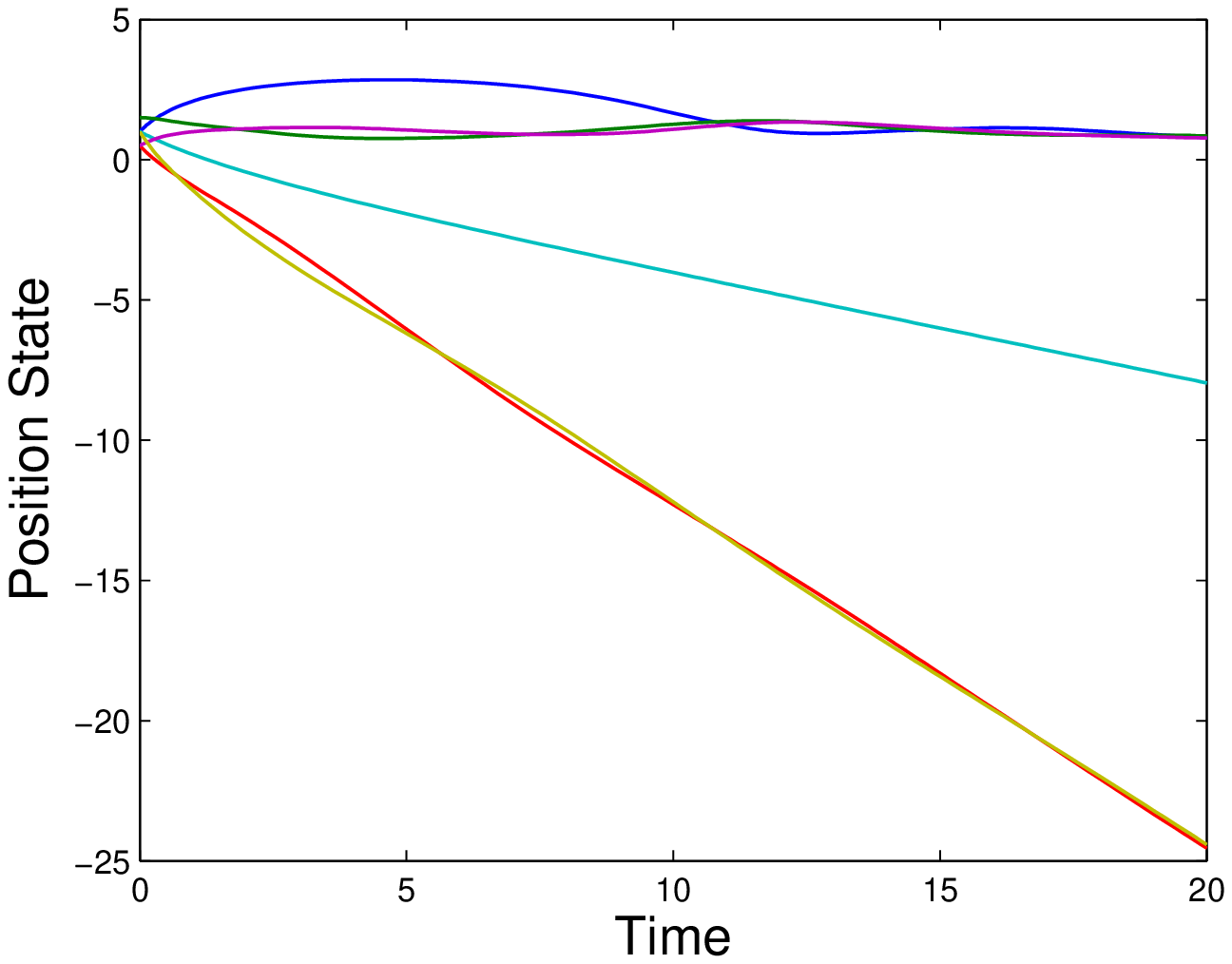}
\includegraphics[width=0.4\textwidth]{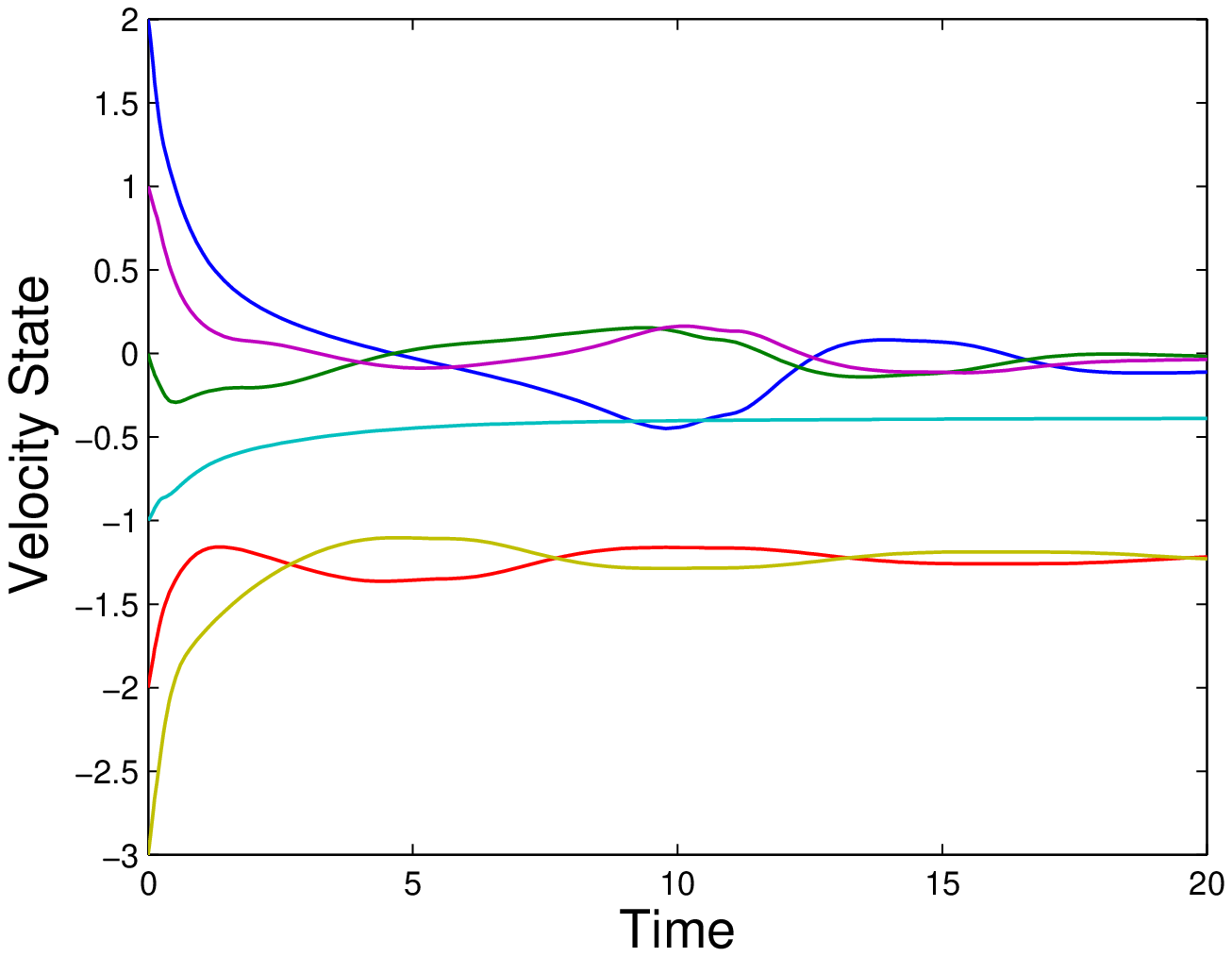}
\caption{Agents with dynamics (\ref{ct2}) and protocol (\ref{cu22}), $H=1$.}\label{fig cu22f}
\end{figure}

\begin{figure}[htbp]
\centering
\includegraphics[width=0.4\textwidth]{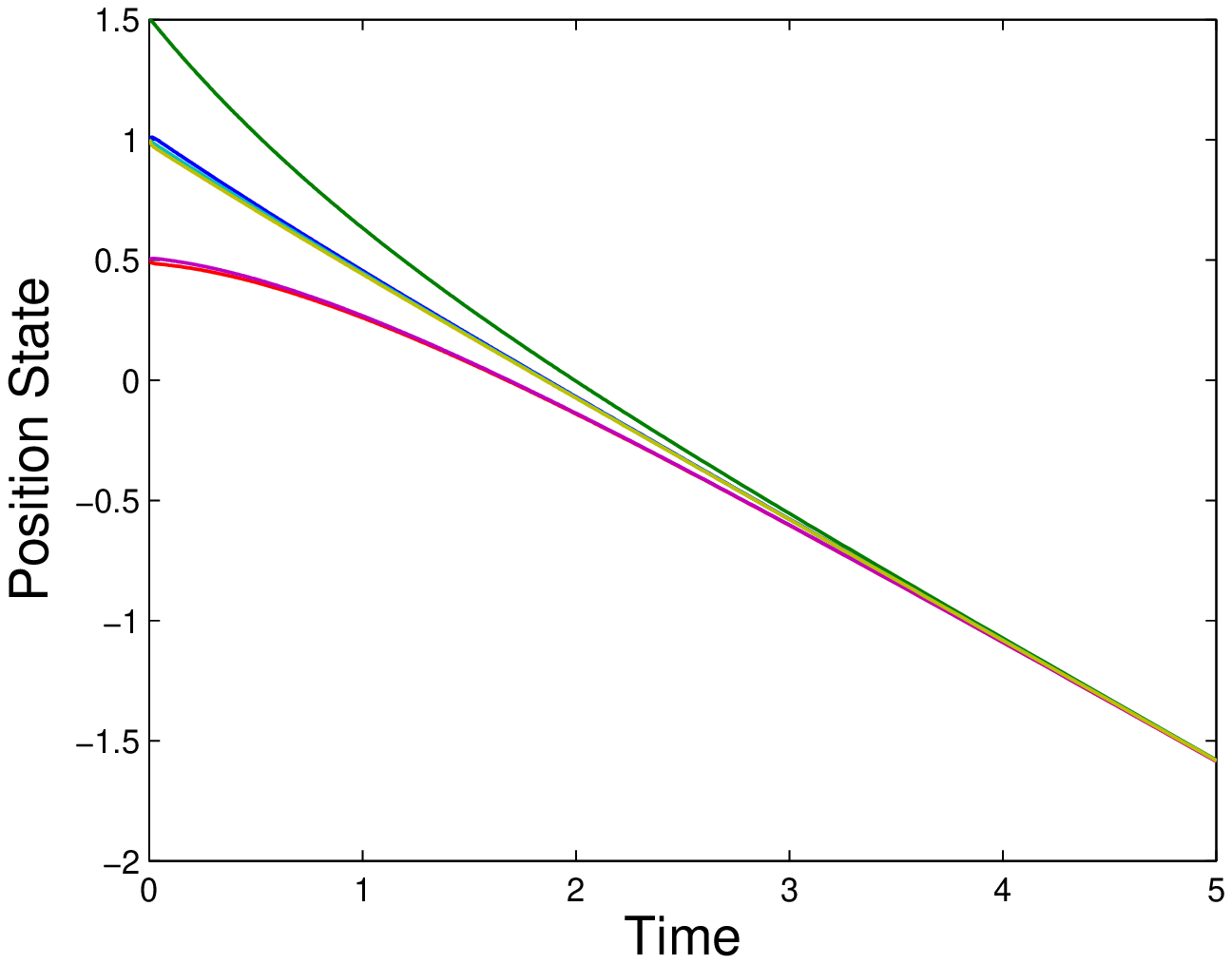}
\includegraphics[width=0.4\textwidth]{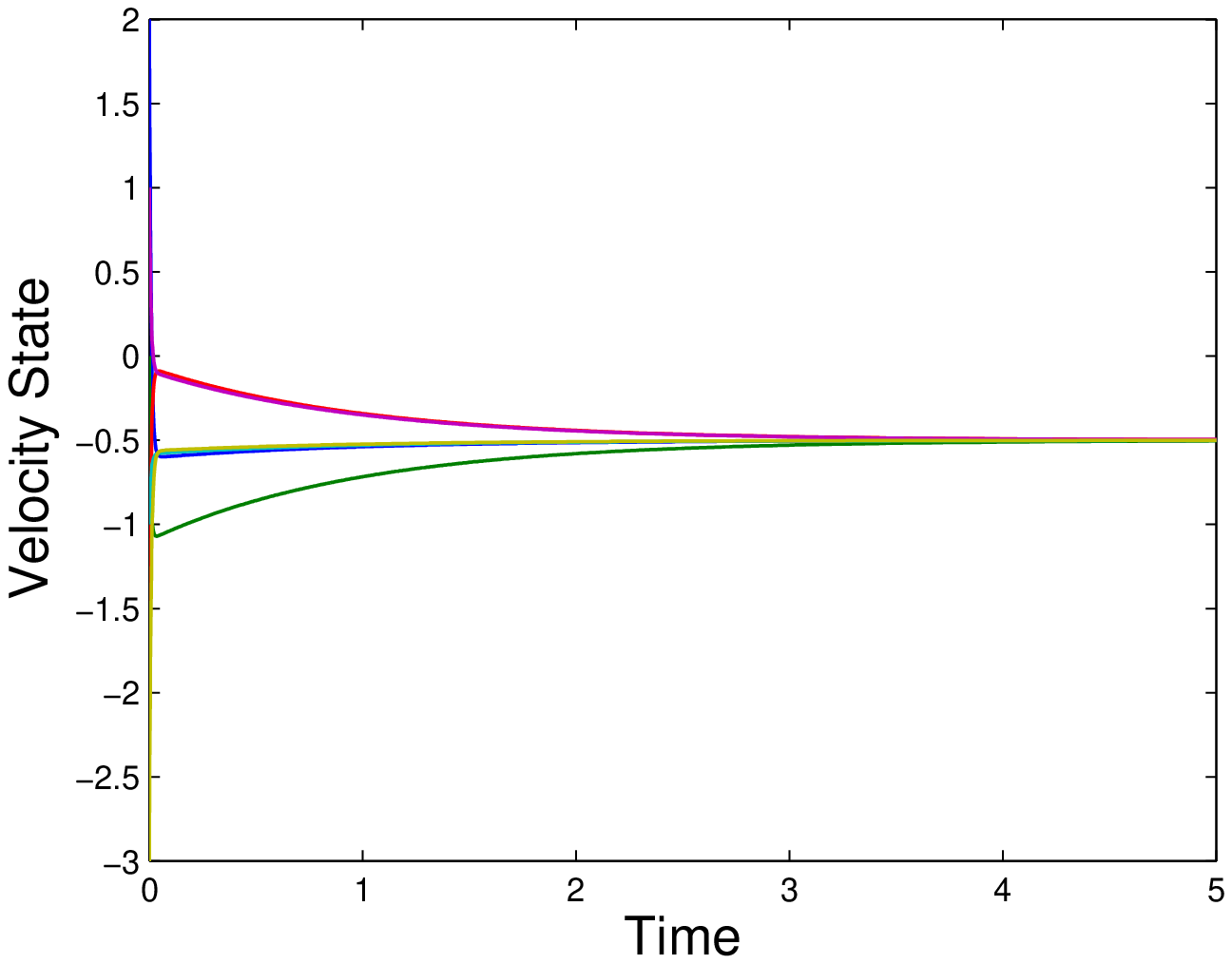}
\caption{Agents with dynamics (\ref{ct2}) and protocol (\ref{cu22}), $H=150$.}\label{fig cu22s}
\end{figure}

\subsection{Applications to Opinion Dynamics}

In this section, we consider the consensus problem of opinion formation among a group of agents. In detail, each agent keeps a real number as its opinion and updates it by taking a weighted average for the opinions of its neighbors. Two agents are called neighbors if their opinions keep a distance less than a constant(called by the confidence bound). Moreover, the weights may change with the evolution of the opinions. In the following, we will see that under a specified opinion-dependent dynamics, what kind of initial profiles can lead to a consensus.

For continuous-time agents, the following smoothed model is considered:
\begin{equation}\label{smooth}
\dot{x}_i=\sum_{i\in\mathcal{V}}\alpha(||x_i-x_j||^2)(x_j-x_i),
\end{equation}
where
\begin{equation}\label{c alpha}
\alpha(s)=\begin{cases}
c, & 0\leq s< (R-\varepsilon)^2,\\
f(s), & (R-\varepsilon)^2\leq s<R^2,\\
0, & s\geq R^2.
\end{cases}
\end{equation}
$x_i\in\mathbb{R}$ denotes the opinion of agent $i$, $c>0$ is the communication weight between neighbors, $R>0$ is the bound of confidence, $f(s)$ is a nonincreasing and Lipschitz continuous function of $s$ in $[(R-\varepsilon)^2, R^2]$, and $f((R-\varepsilon)^2)=c$, $f(R^2)=0$. This smoothed model makes such an assumption that when the opinion of agent $j$ is running out of the confidence bound of agent $i$, the information transmission between them vanishes smoothly. In \cite{Ceragioli12}, $\varepsilon$ is set by a sequence which $f(s)$ closely depends on, $i.e.$, $f(s)=\frac{c}{\varepsilon}(R-\sqrt{s})$, this model is called an $\varepsilon$ approximation for H-K model. It is obvious that Theorem \ref{th ccu1} can be applied to this model. Therefore, the average consensus can be reached if the initial states of agents satisfy (\ref{ccu1con}).

In fact, if the initial opinions are symmetrically distributed, we can obtain a more relaxed condition.

Consider a system consisting of $n$ agents, agent $i$ keeps a real number $x_i$ as its opinion. Assume that $x_i\leq x_j$ if $i\leq j$. We say the states are symmetrically distributed if there exists a real number $x_0$, such that $x_0=\frac{x_i+x_j}{2}$ for any $i+j=n$. We present the following proposition, the relevant proof is presented in Section 7.
\begin{proposition}\label{pr 2n-3}
Consider model (\ref{smooth}) with $n\geq4$ agents, suppose the initial states of the agents are symmetrically distributed. For any $t>0$, if the communication graph is disconnected, there are at least $2n-3$ pairs of disconnected agents.
\end{proposition}
\begin{theorem}\label{th ctsymmetry}
Consider model (\ref{smooth}) with $n$ symmetrically distributed opinions in the initial time. Then the following statements hold.

(i). For $2\leq n\leq3$, the average consensus of the opinions is achieved if and only if the initial communication graph is connected.

(ii). For $n\geq4$, the average consensus of the opinions is achieved if the following inequality holds:
\begin{equation}\label{ct symmetric}
\frac{1}{2} \sum_{i\in\mathcal{V}} \sum_{j\in\mathcal{V}} \int_0^{||x_i(0)-x_j(0)||^2}\alpha(s)ds< (2n-3)\int_0^{R^2}\alpha(s)ds.
\end{equation}
\end{theorem}
\begin{proof} (i). From the analysis in the proof of Theorem \ref{th ccu1}, we know that preserving the connectivity of the communication graph is the key to make the agents reach consensus.

For $n=2$. We let $x_1$ and $x_2$ be the two agents' opinions and $e=x_2-x_1$. Then $\dot{e}=-2\alpha_{12}e$. Sufficiency: Note that $\dot{e}\geq0$ if $e<0$ and $\dot{e}\leq0$ if $e>0$, which in turn implies that $|e|$ is decreasing of $t$, together with $|e(0)|<R$, we have $|e|<R$ for any $t\geq0$. Necessity: Suppose that $|e(0)|\geq R$, then $\dot{e}=0$, consensus will never be reached.

For $n=3$. Let $x_1$, $x_2$, $x_3$ be the three opinions and $x_1\leq x_2\leq x_3$. From Lemma \ref{pr symmetry}, we have $x_2=\frac{x_1+x_3}{2}$ and $\dot{x}_2=0$ for any $t\geq0$. Then $\dot{x}_1=\alpha_{12}(x_2-x_1)+\alpha_{13}(x_3-x_1) =(\alpha_{12}+2\alpha_{13})(x_2-x_1)$. Similarly, we have $\dot{x}_3=(\alpha_{23}+2\alpha_{13})(x_2-x_3)$. Let $e_1=x_1-x_2$, $e_2=x_3-x_2$, it follows that $\dot{e}_1=-(\alpha_{12}+2\alpha_{13})e_1$, $\dot{e}_2=-(\alpha_{23}+2\alpha_{13})e_2$. Sufficiency: Due to the fact that $|e_1(0)|<R$, $|e_2(0)|<R$, we obtain that $|e_1(t)|<R$ and $|e_2(t)|<R$ for any $t\geq0$. That is, the connectivity of the communication graph is maintained. Necessity: Suppose the initial communication graph is not connected. If $|e_1(0)|>R$, then $\dot{x}_1=0$, together with $\dot{x}_2=0$, one has $\dot{e}_1=0$, a contradiction. If $|e_2(0)|>R$, then $\dot{x}_3=0$, together with $\dot{x}_2=0$, the consensus cannot be reached.

(ii). By employing Proposition \ref{pr 2n-3}, the proof is similar to the one of Theorem \ref{th ccu1}.
\end{proof}

Now we consider an example of the smoothed opinion dynamics (\ref{smooth}). Suppose the system consists of $20$ evenly distributed opinions in the initial time. Let $R=1$, $\varepsilon=0.1$, $c=1$, which implies that $f(s)=10(1-\sqrt{s})$, the distance between adjacent agents is set as $d=0.2$. It can be calculated that (\ref{ct symmetric}) cannot be satisfied. Fig. \ref{fig smoothf} shows the evolution of all the opinions and the variation of the Lyapunov function (\ref{smooth W}). If we change $d$ to be $0.05$, (\ref{ct symmetric}) can be guaranteed. The average consensus is achieved, and (\ref{smooth W}) gradually vanishes, as shown in Fig. \ref{fig smooths}.
\begin{figure}
\centering
\includegraphics[width=0.4\textwidth]{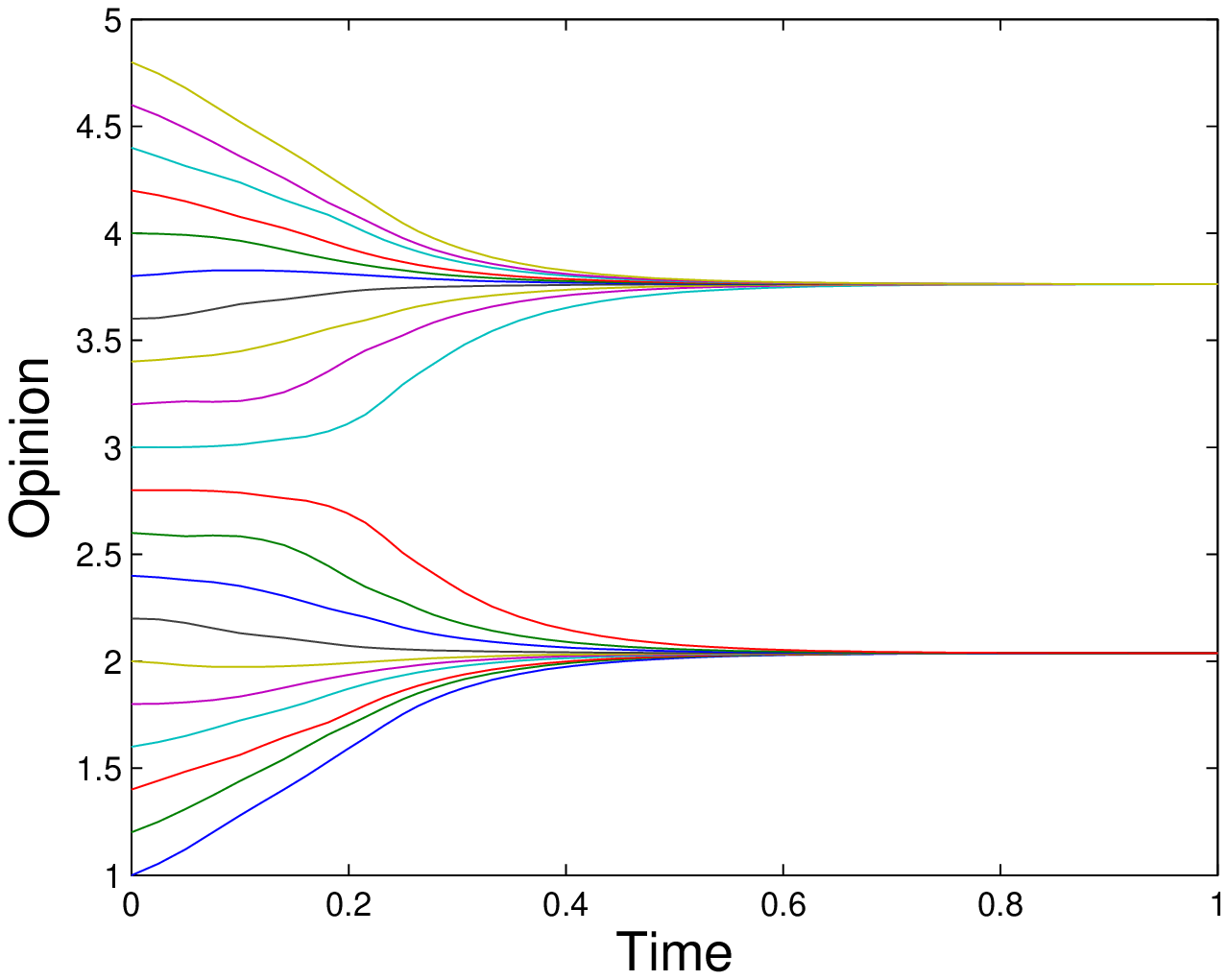}
\includegraphics[width=0.4\textwidth]{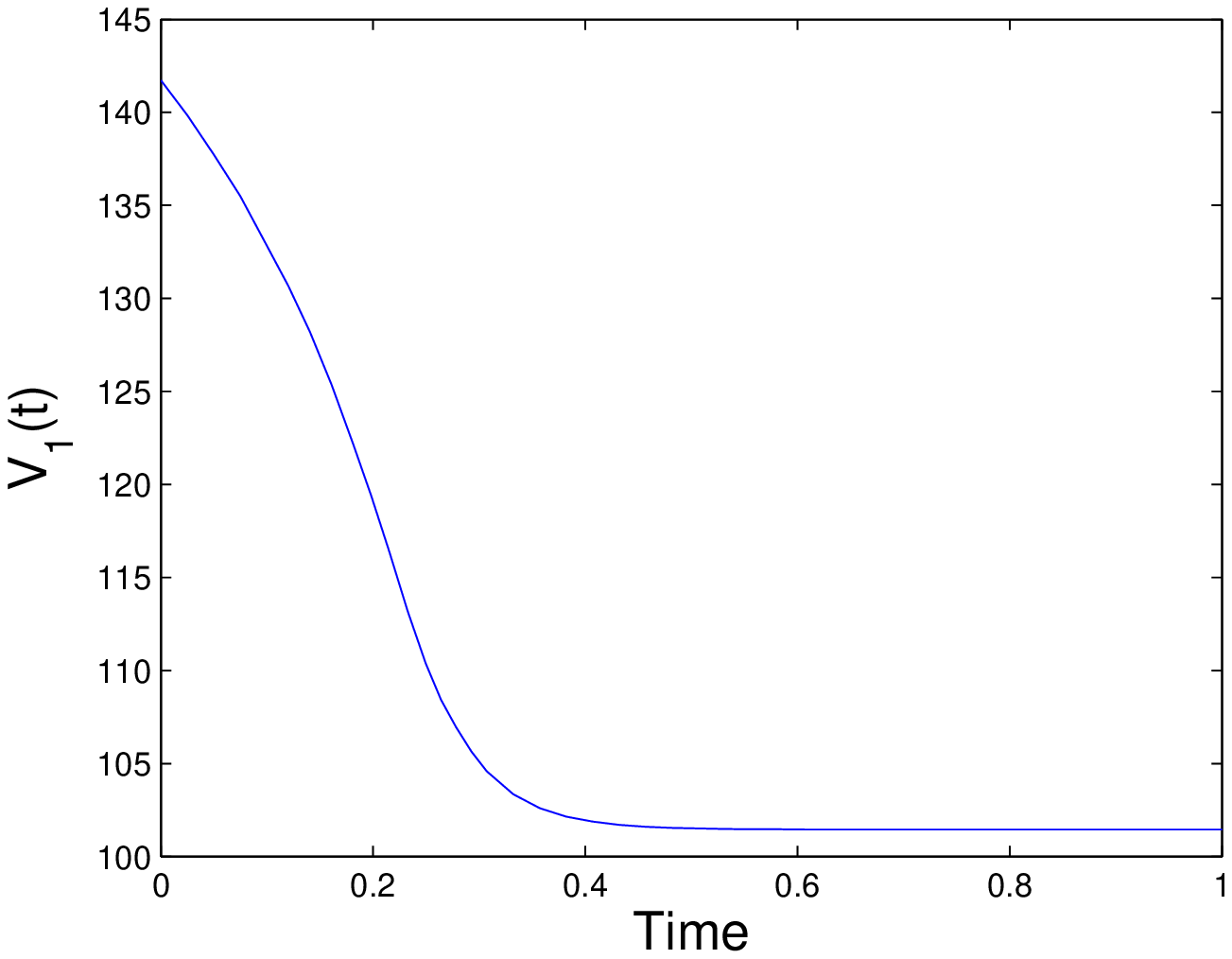}
\caption{Model (\ref{smooth}) with $d=0.2$.}\label{fig smoothf}
\end{figure}
\begin{figure}
\centering
\includegraphics[width=0.4\textwidth]{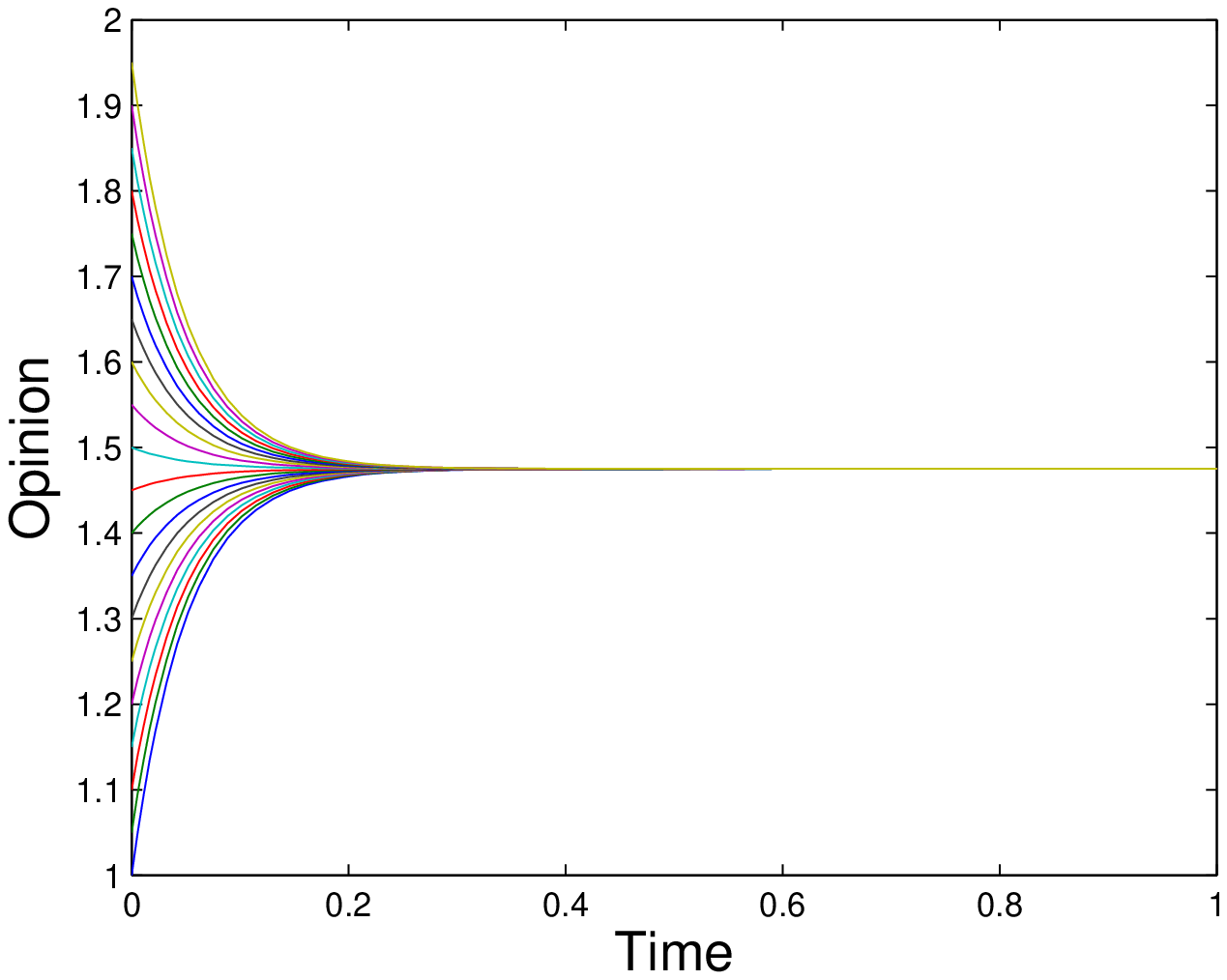}
\includegraphics[width=0.4\textwidth]{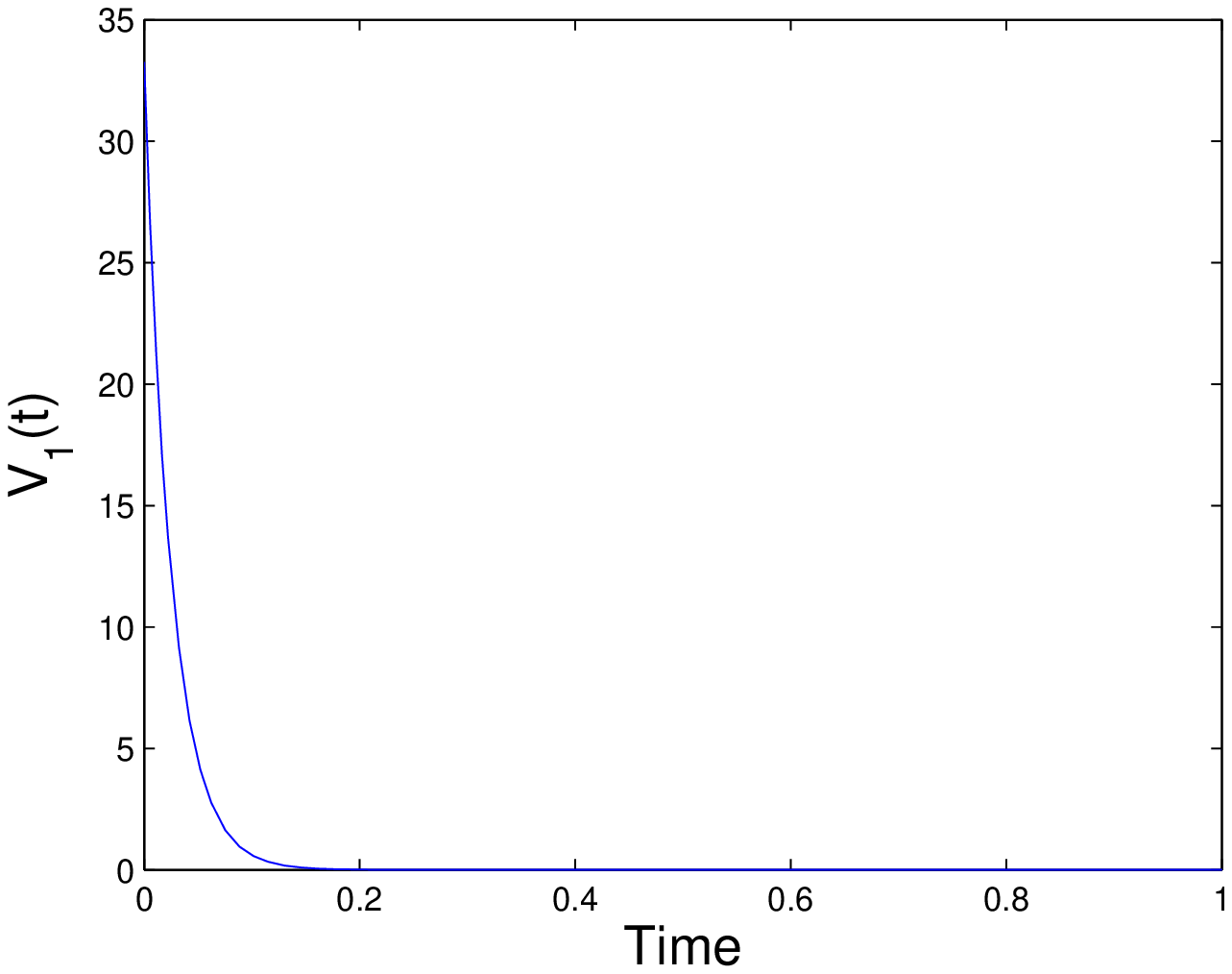}
\caption{Model (\ref{smooth}) with $d=0.05$.}\label{fig smooths}
\end{figure}

For discrete-time opinion dynamics, the following opinion evolution model is considered:
\begin{equation}\label{dt opinion}
x_i(t+1)=\sum_{j\in\mathcal{V}}w_{ij}(x)x_j(t).
\end{equation}
where $w_{ij}\geq0$ denotes the weight between agent $i$ and agent $j$, and $\sum_{j\in\mathcal{V}}w_{ij}=1$ for any $i\in\mathcal{V}$. When the system starts running, each agent will take those agents into account whose opinions differ from its own not more than the confidence bound $R>0$. We make an assumption that each agent employs the same weight $i.e.$, $h>0$ when it considers its neighbors except itself. Since the agent will consider its own opinion in a positive way, to make this hold, we assume $(n-1)h<1$. Then model (\ref{dt opinion}) can be rewritten by
\begin{equation}\label{dt opinion1}
x_i(t+1)=(1-h\sum_{j\neq i}\alpha_{ij})x_i(t)+ h\sum_{j\neq i}\alpha_{ij}x_j(t),
\end{equation}
where
\begin{equation}\label{d alpha}
\alpha(s)=\begin{cases}
1, & 0\leq s<R^2;\\
0, & s\geq R^2.
\end{cases}
\end{equation}
Then (\ref{dt opinion}) is equivalent to (\ref{dt1}) with (\ref{ddu1}). Since $\alpha(\cdot)$ is nonincreasing and $h<\frac{1}{n-1}$, Theorem \ref{th ddu1} can be employed. The agents will achieve the average consensus of opinions if (\ref{ddu1con}) holds.

Similar to Theorem \ref{th ctsymmetry}, the following results for discrete-time opinion dynamics are valid, we omit the corresponding proof due to its simpleness.
\begin{theorem}\label{th dtsymmetry}
Consider model (\ref{dt opinion1}) with $n$ symmetrically distributed opinions in the initial time and $h<\frac{1}{\alpha(0)(n-1)}$. Then the following statements hold.

(i). For $2\leq n\leq3$, the average consensus of the opinions is achieved if and only if the initial communication graph is connected.

(ii). For $n\geq4$, the average consensus of the opinions is achieved if there exists an $r\in[0,R^2)$, such that
\begin{equation}\label{dt symmetric}
W(0)<(2n-3)w(R^2).
\end{equation}
\end{theorem}
Consider model (\ref{dt opinion1}) with $15$ evenly distributed opinions in the initial time. The distance between adjacent agents is $d=0.35$. Set $R=1$, $r=0.1$, $h=\frac{1}{n}$, the initial states do not satisfy (\ref{dt symmetric}). Fig. \ref{fig dopinionf} describes the evolution of opinions and $W(t)$, we can observe that the opinions fail to reach consensus. When we set $d=0.08$, (\ref{dt symmetric}) is valid for $r=0.1$. The average consensus is reached, as shown in Fig. \ref{fig dopinions}.

\begin{figure}[htbp]
\centering
\includegraphics[width=0.4\textwidth]{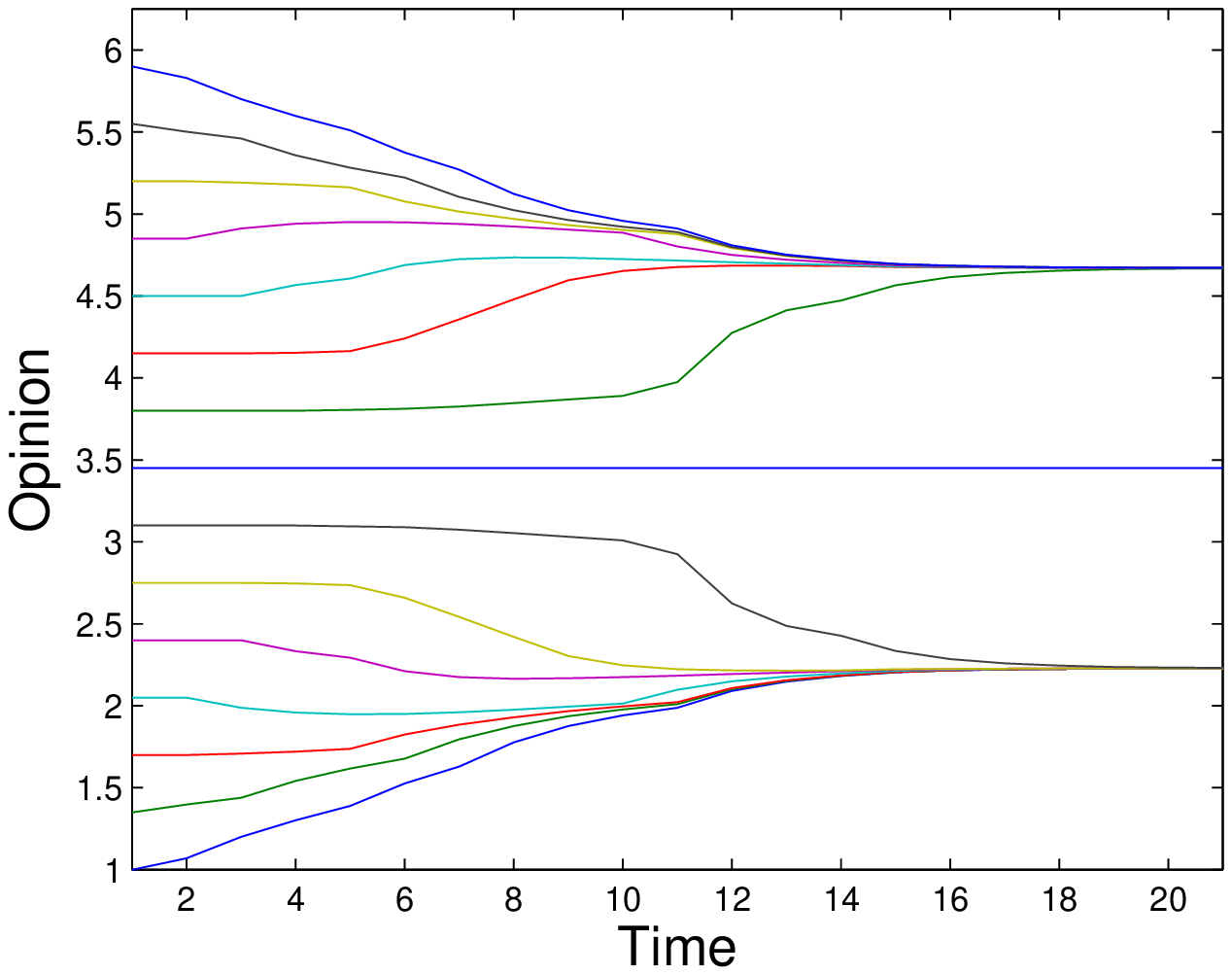}
\includegraphics[width=0.4\textwidth]{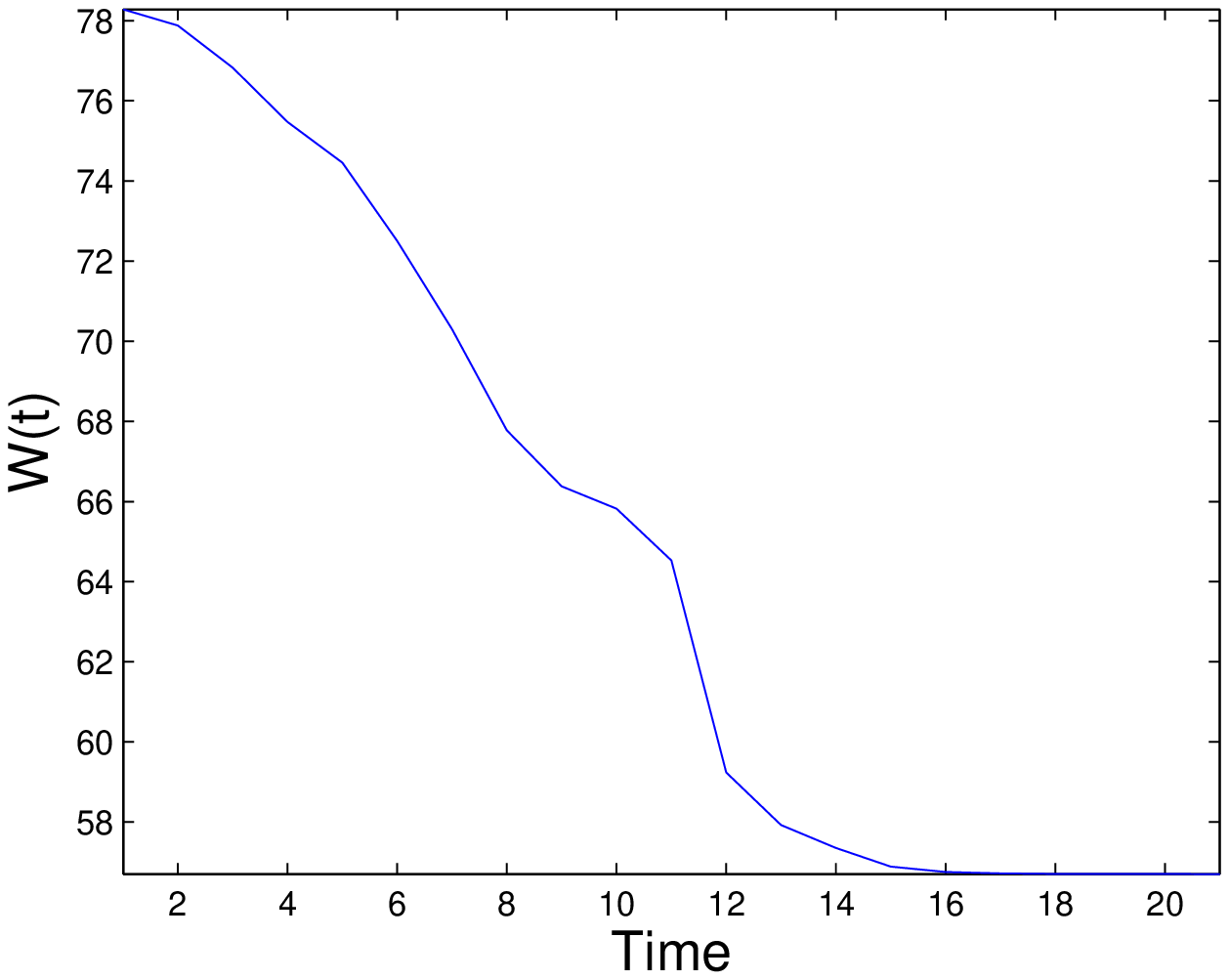}
\caption{Model (\ref{dt opinion}) with $d=0.35$.}\label{fig dopinionf}
\end{figure}

\begin{figure}[htbp]
\centering
\includegraphics[width=0.4\textwidth]{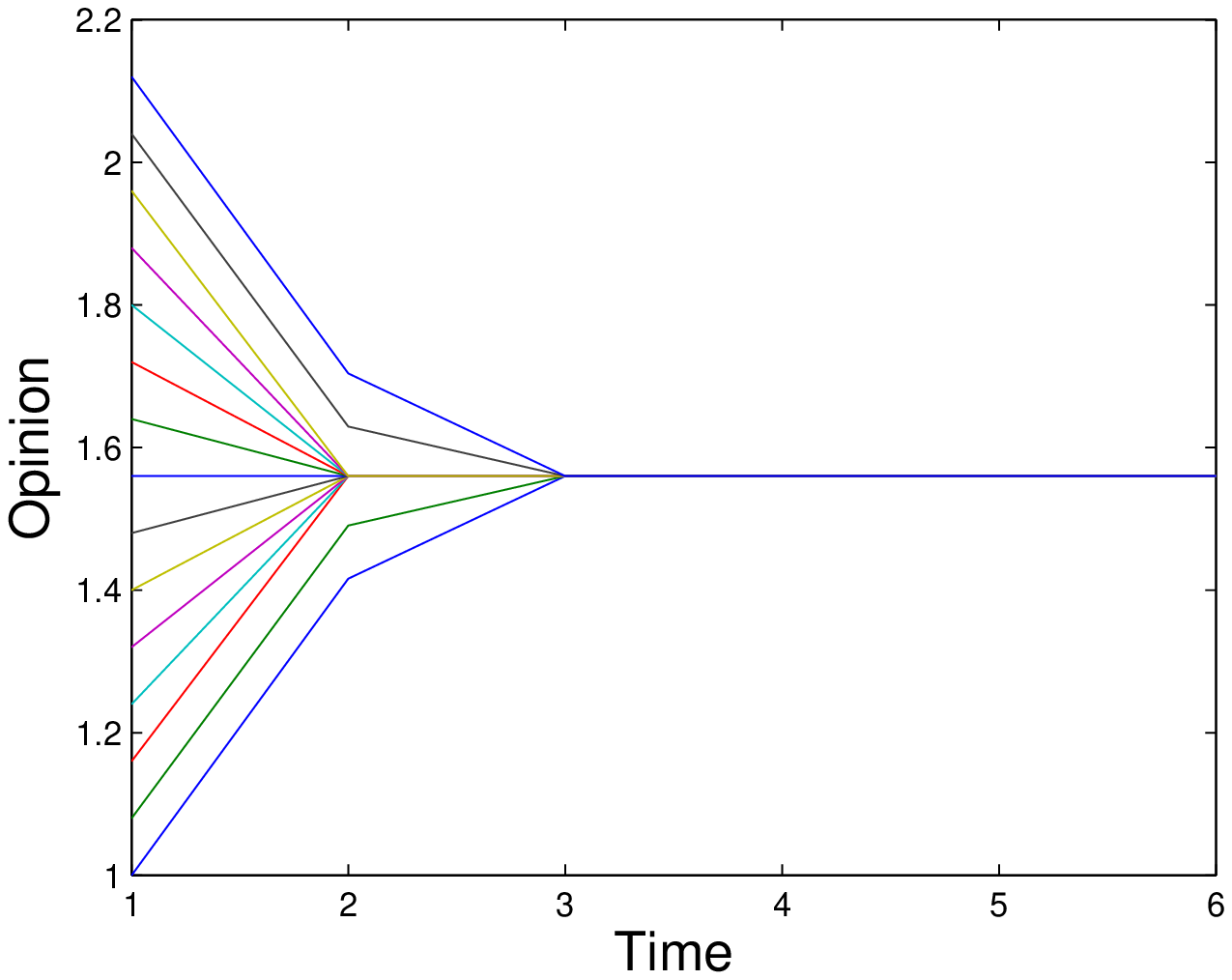}
\includegraphics[width=0.4\textwidth]{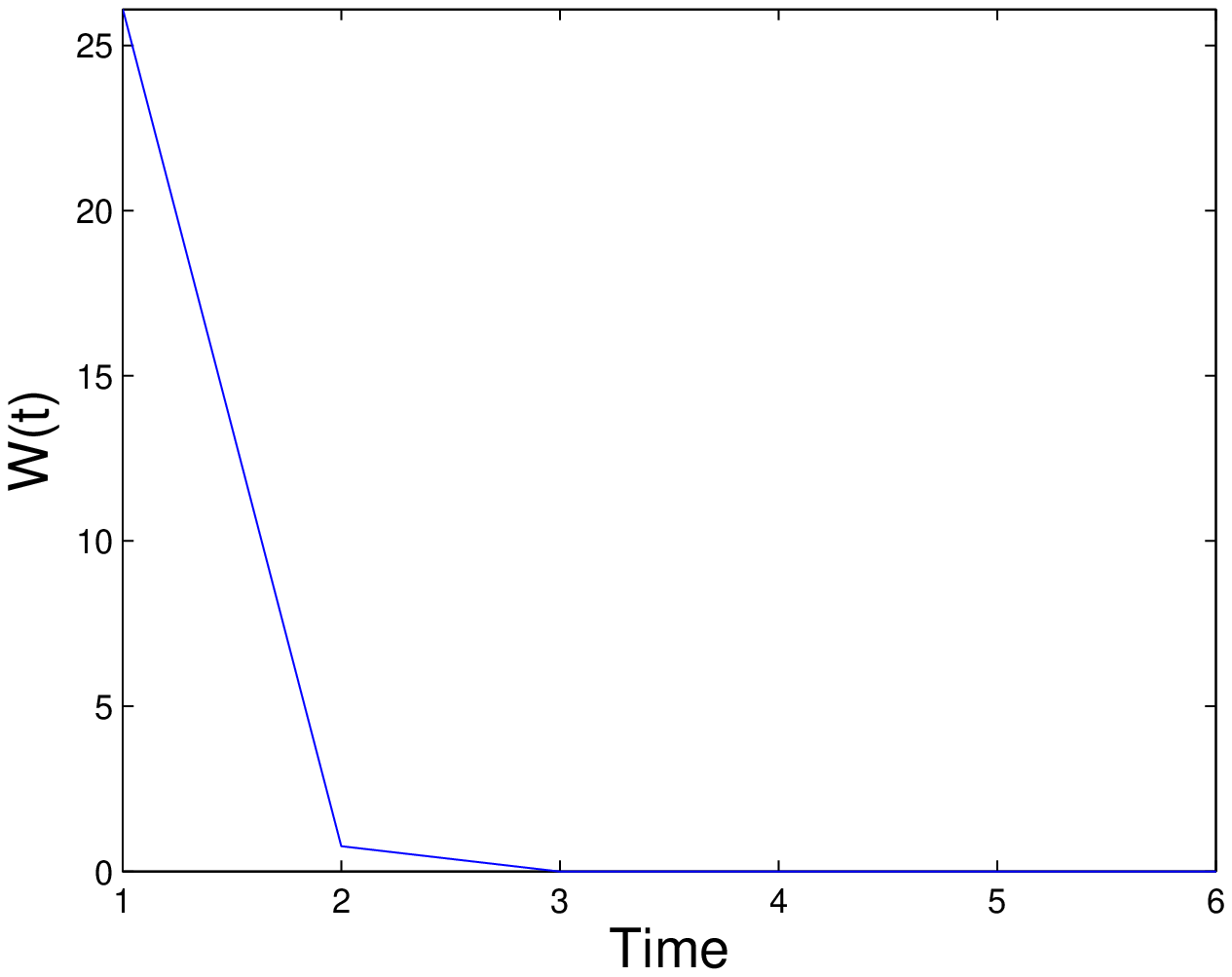}
\caption{Model (\ref{dt opinion}) with $d=0.08$.}\label{fig dopinions}
\end{figure}

It is easy to see that Theorem \ref{th ctsymmetry} and Theorem \ref{th dtsymmetry} also hold when $\alpha(\cdot)$ just satisfies Assumption 2 and Assumption 4, respectively. Because the corresponding proof does not require a particular $\alpha(\cdot)$. In order to verify that taking a different $r$ is helpful to satisfy the initial condition, we give an example in the following.

Now we consider model (\ref{dt opinion1}) with a varying communication weight. Assume that there are $20$ evenly distributed opinions in the initial time, the communication weight between agents decays when their opinion difference increases. Let $R=1.5$, $\alpha(s)=-10s+25$, $h=\frac{1}{\alpha(0)n}$. It is found that when we set $d=0.07$, (\ref{dt symmetric}) hold for $r=1.8$ but it does not hold for $r=0$. Fig. \ref{fig dopinions1} shows the result for $d=0.07$.
\begin{figure}[htbp]
\centering
\includegraphics[width=0.4\textwidth]{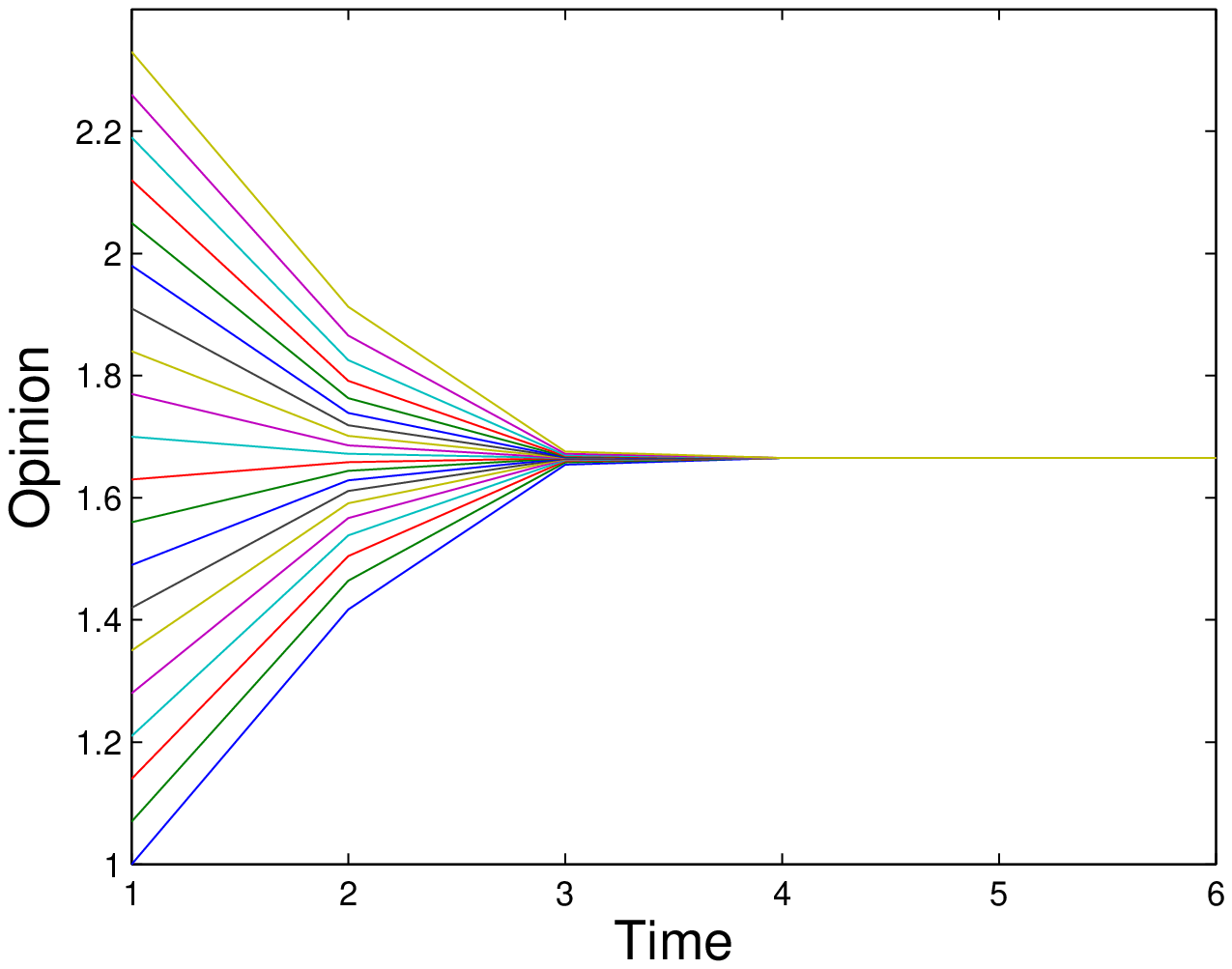}
\includegraphics[width=0.4\textwidth]{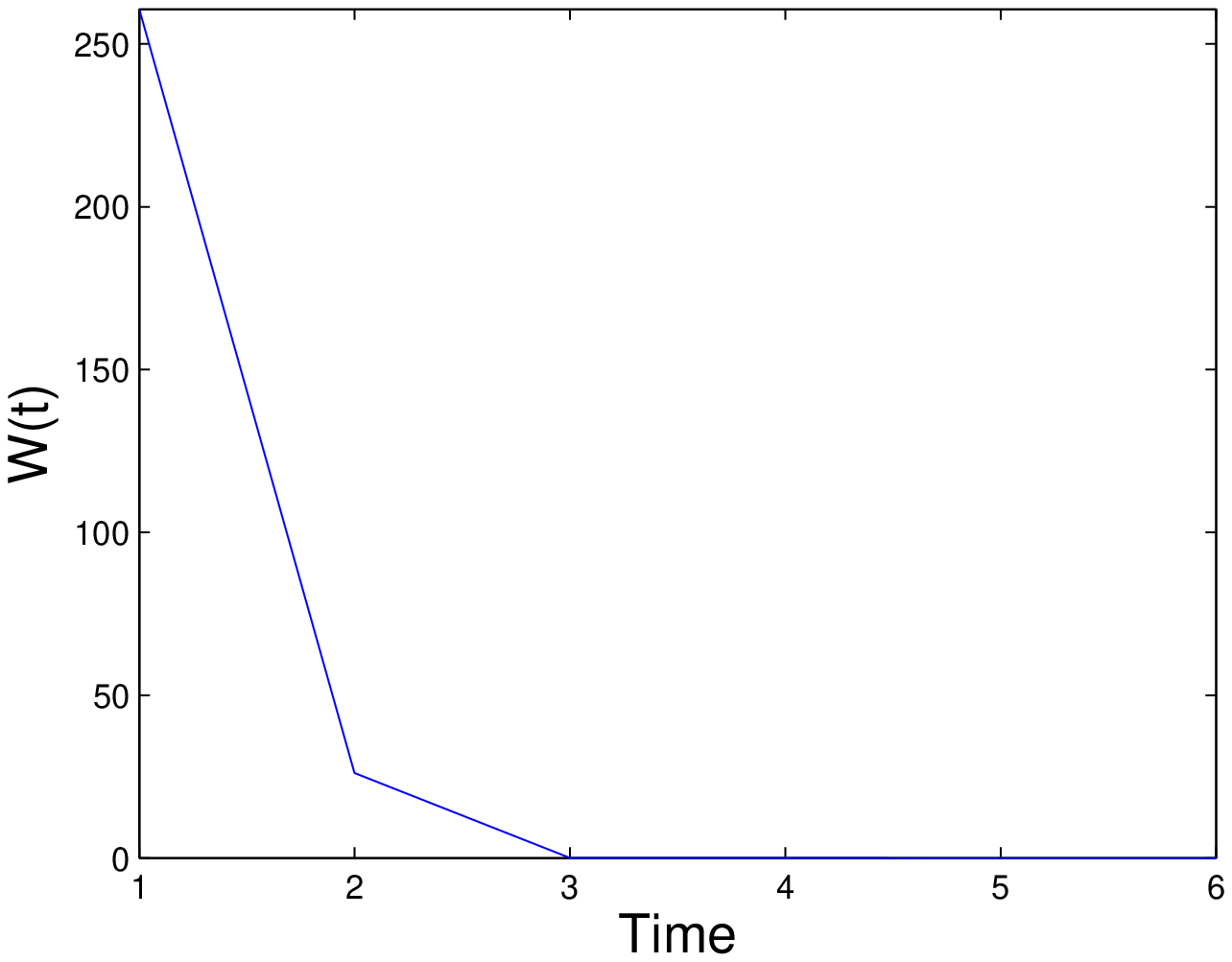}
\caption{Model (\ref{dt opinion}) with $d=0.07$.}\label{fig dopinions1}
\end{figure}

\subsection{Applications to Rendezvous}

Now we consider the rendezvous problem of multiple agents with continuous-time dynamics and discrete-time dynamics. In such problems, some communication links may be lost due to the moving of the agents and therefore the rendezvous will not be realized \cite{Lin071,Lin072}. Unlike the study in \cite{Su10}, we do not employ potential functions to preserve the connectivity of the network. What we mainly concern about is that under what kind of initial states the network can be always connected. In the following, model (\ref{ct2}) with (\ref{ccu2}) and model (\ref{dt2}) with (\ref{ddu2}) will be applied to solve the rendezvous problem, several simulations are represented. In the simulations, the red point denotes the initial state of an agent and the blue point is its final state. The lines in different colors denote the trajectory of the agents.

For continuous-time systems with dynamics (\ref{ct2}), suppose there is a system consisting of $6$ agents. With protocol (\ref{ccu2}), all the agents move in the plane and employ (\ref{c alpha}) as the transmission weight. In Fig. \ref{cfrendezvous}, the rendezvous fails since the connectivity of the communication network is broken during the agents' moving. We can see that even if the consensus of the agents' position states is not reached, the velocity of all the agents still vanish to zero in the end. Under condition (\ref{ccu2con}), Fig. \ref{csrendezvous} shows that the rendezvous problem is solved.
\begin{figure}[htbp]
\centering
\includegraphics[width=0.4\textwidth]{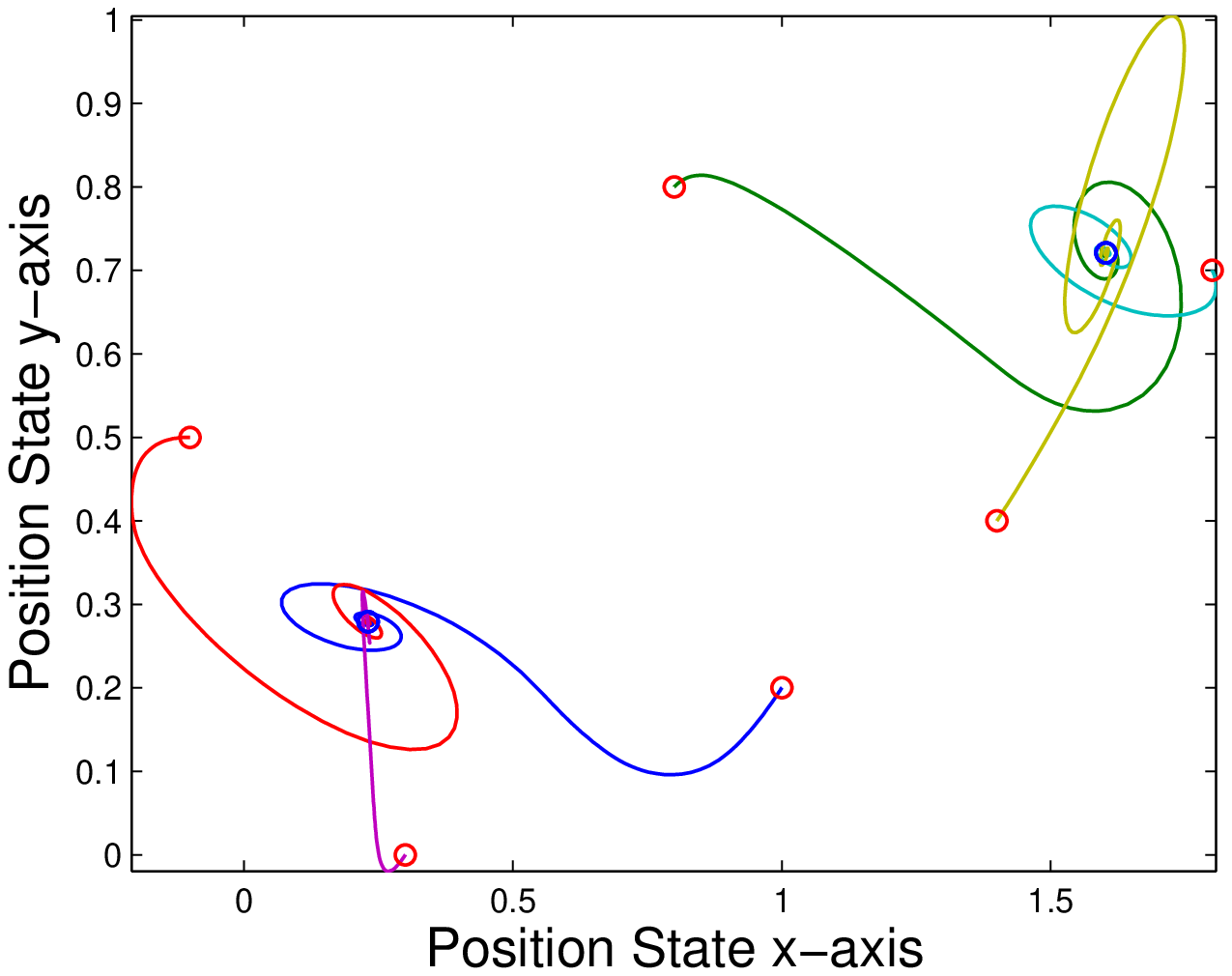}
\includegraphics[width=0.4\textwidth]{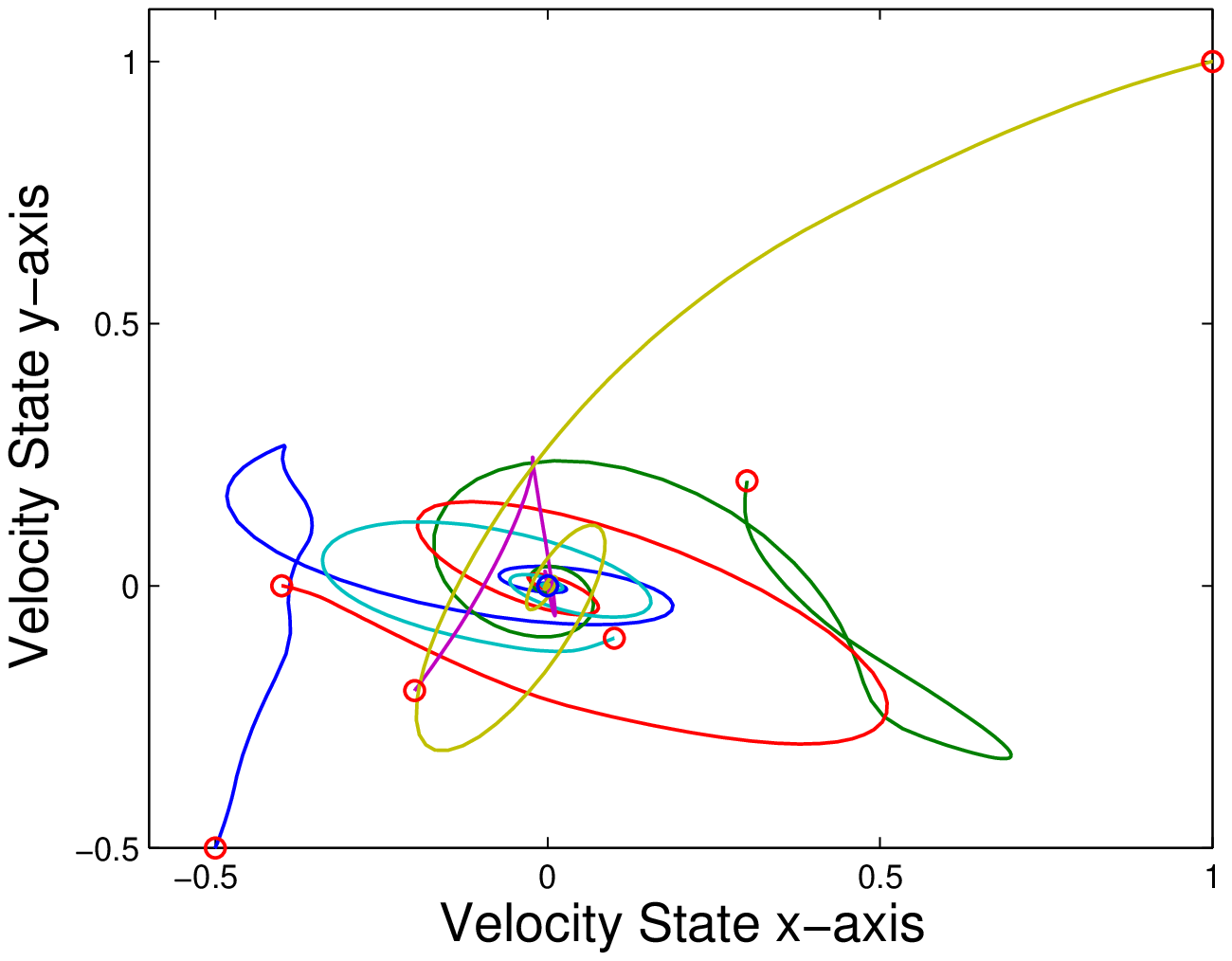}
\caption{Protocol (\ref{ccu2}) fail to solve the rendezvous problem of agents with dynamics (\ref{ct2}).}\label{cfrendezvous}
\end{figure}
\begin{figure}[htbp]
\centering
\includegraphics[width=0.4\textwidth]{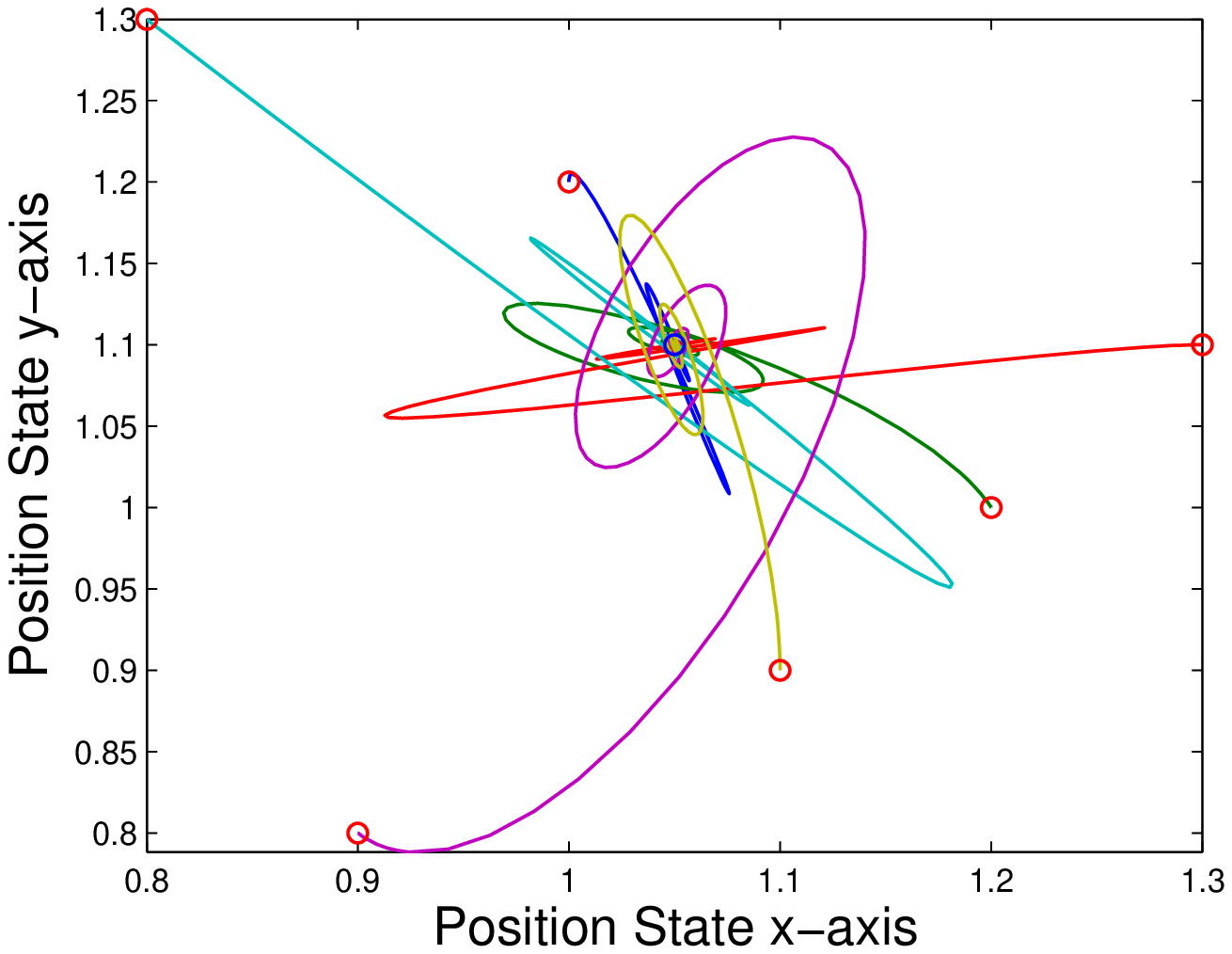}
\includegraphics[width=0.4\textwidth]{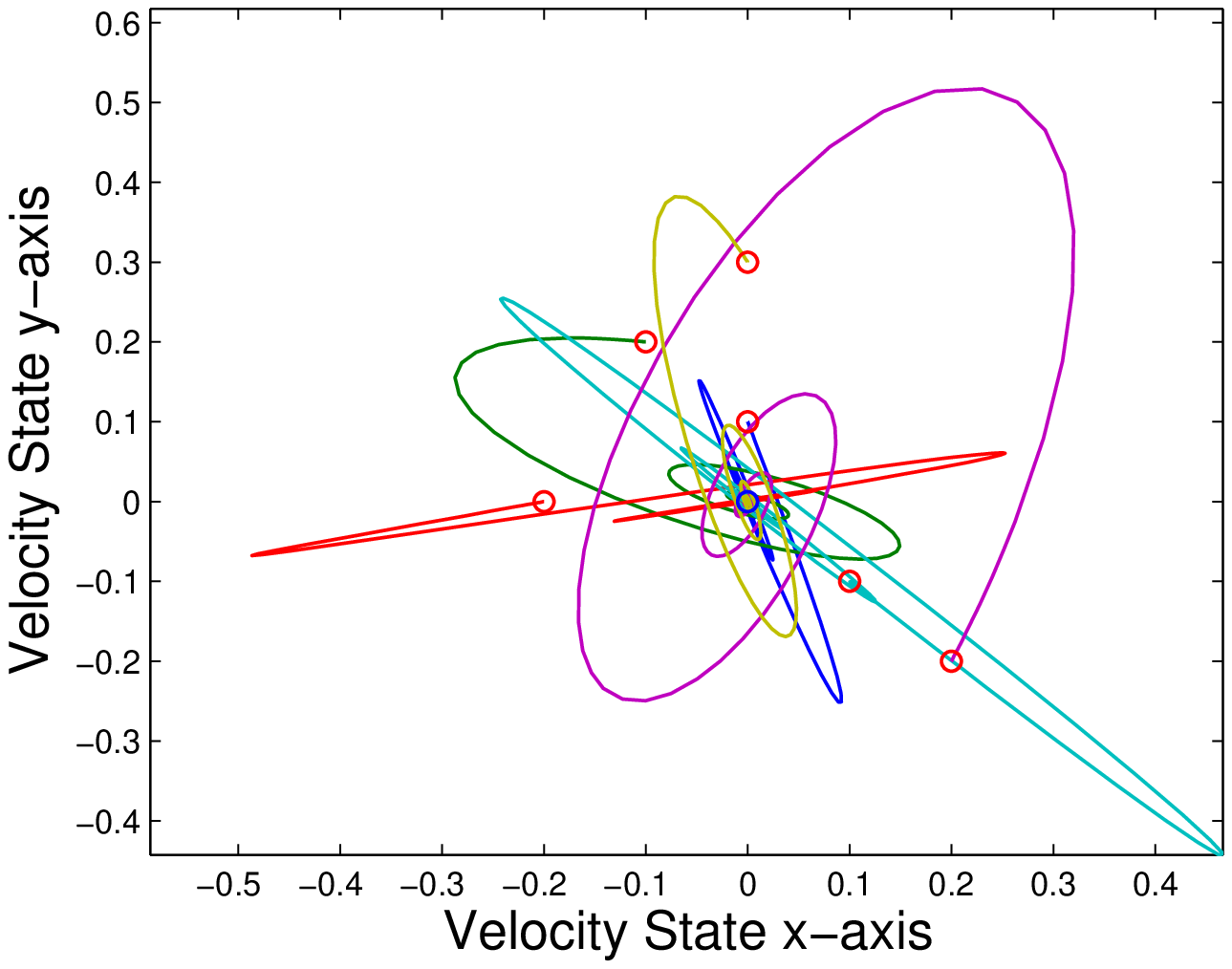}
\caption{Protocol (\ref{ccu2}) success to solve the rendezvous problem of agents with dynamics (\ref{ct2}).}\label{csrendezvous}
\end{figure}
For discrete-time systems with dynamics (\ref{dt2}), consider a system consisting of $6$ agents. Applying protocol (\ref{ddu2}) with (\ref{d alpha}) as the communication weight. Let $h_1=1$, $h_2=1.5$, $h_3=0.14$, then (\ref{d2c1}) and (\ref{d21c2}) are satisfied. When the initial states of all the agents are restricted by (\ref{ddu2con}) with $r=0.1$, Fig. \ref{dsrendezvous} shows that the rendezvous is reached.
\begin{figure}[htbp]
\centering
\includegraphics[width=0.4\textwidth]{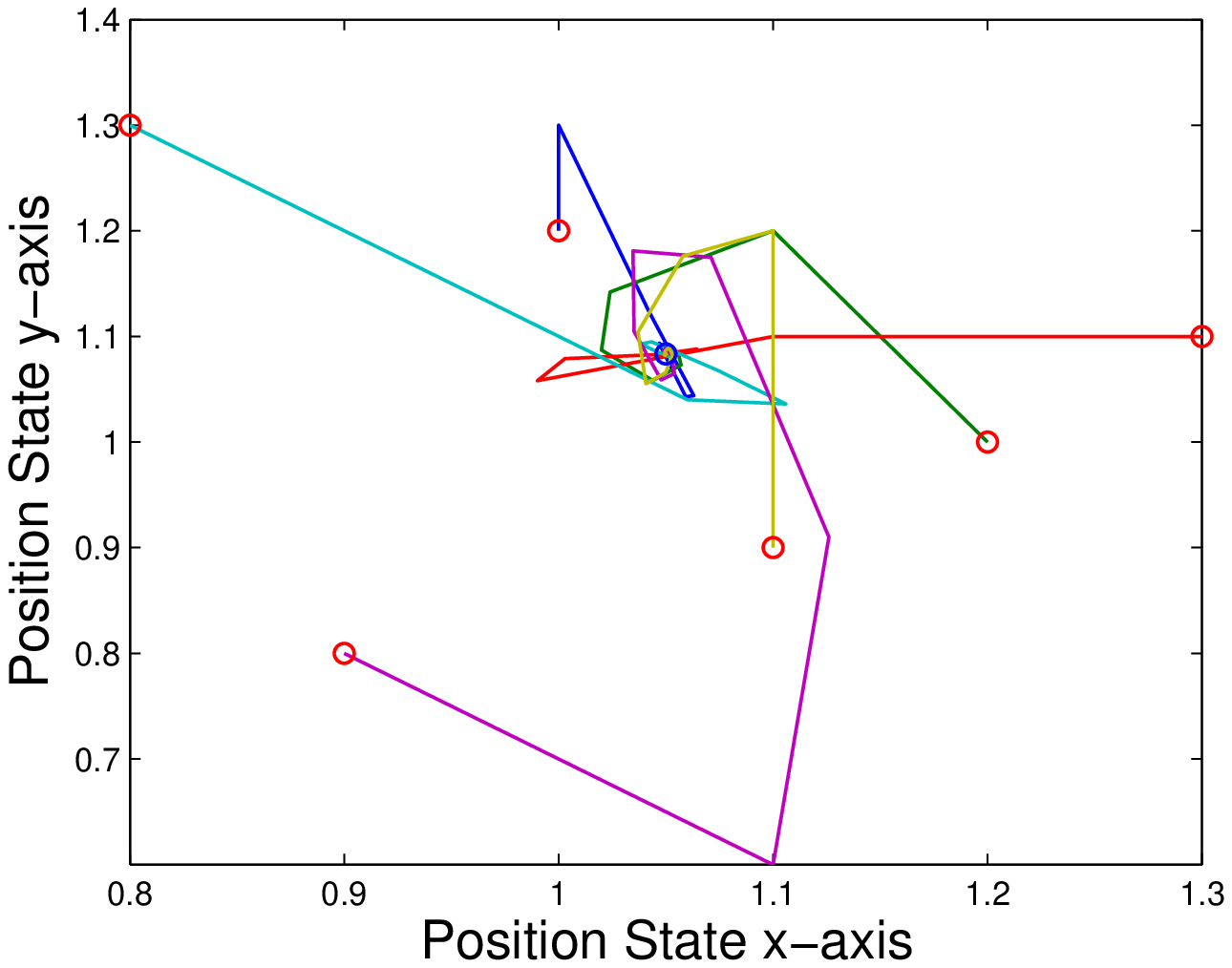}
\includegraphics[width=0.4\textwidth]{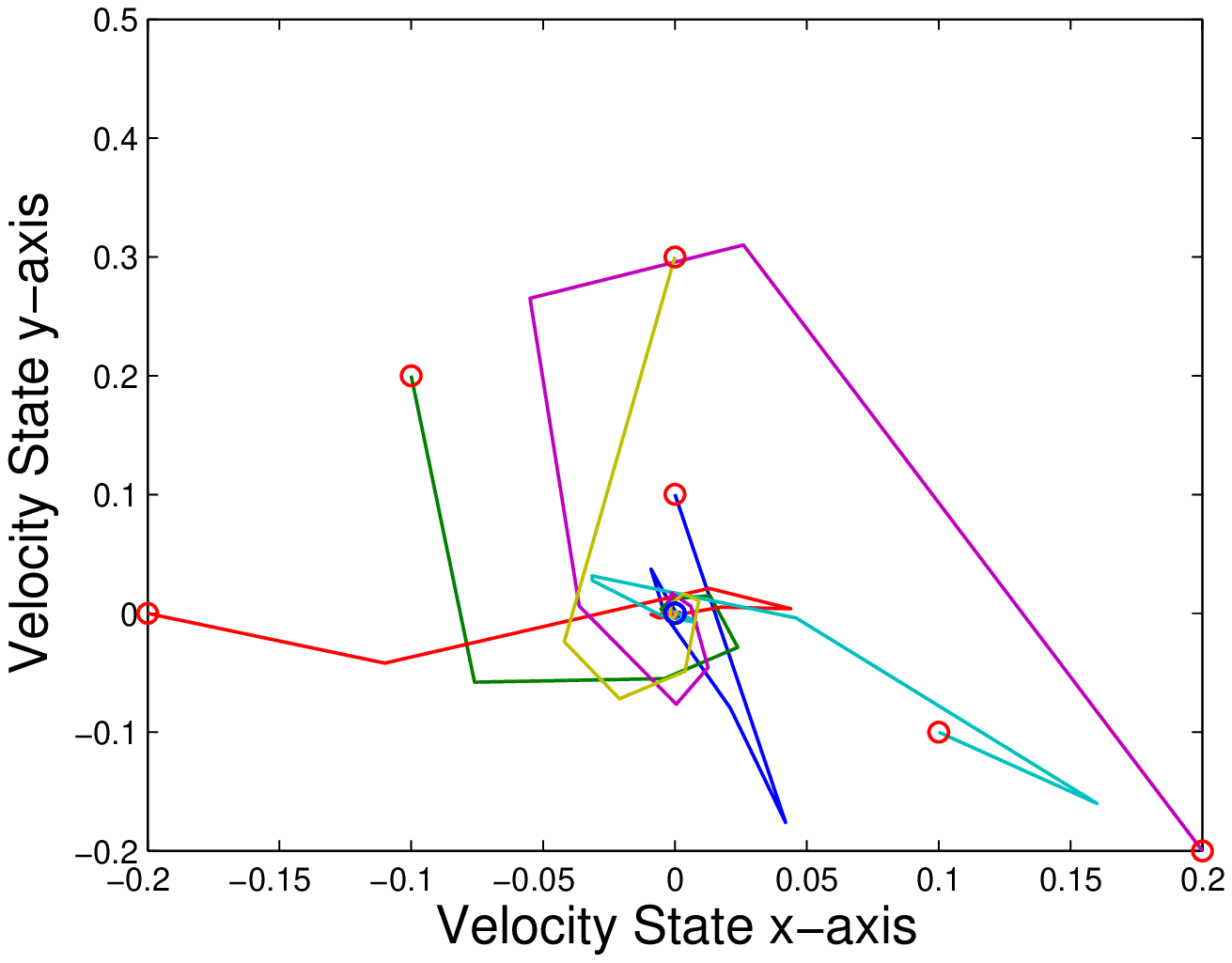}
\caption{Protocol (\ref{ddu2}) success to solve the rendezvous problem of agents with dynamics (\ref{dt2}).}\label{dsrendezvous}
\end{figure}

\section{Conclusion}
In this paper, the consensus problem for two classes of state-dependent switching systems have been considered. The first case describes some systems in networks with fixed connectivity. For these systems, the volume of information in communication varies but always exists as the evolution of the agents. The second one represents some systems whose communication graph is entirely determined by agents' states and thus some interaction links may be lost as the system runs. Under each kind of information transmission, the continuous-time and discrete-time systems have been studied respectively. In networks with fixed connectivity, we have proved that under a connected communication graph, consensus is reached if the state-dependent weight $\alpha(\cdot)$ or the initial configuration of the agents satisfies some conditions. In networks with state-dependent connectivity, consensus would be reached if the initial states of all the agents are under a restriction. The results of these general nonlinear systems have been applied to C-S model, opinion dynamics and rendezvous, the applications have been verified by several simulations.

Nevertheless, all the criterions for consensus are sufficient but not necessary and hence can probably be further relaxed. For example, how to generalize the undirected communication graph to be a directed one in the first case and whether the right hand side of the inequality in initial conditions can be larger. These problems are currently under exploring. Moreover, if $\alpha(\cdot)$ in continuous-time systems is relaxed to be discontinuous, the trajectory of the agents should be considered in the sense of set-valued analysis. A similar result may be obtained by nonsmooth Lynapunov methods.



%
\appendix[Proofs of Several Lemmas and Propositions]

{\it Proof of Lemma \ref{rankL}:} By the definition of $L$, $\mathbf{1}_n$ is always the eigenvector of $L$ associated with zero. Therefore, $M$ is the subspace of the eigenspace of $L\otimes I_m$ corresponding to zero, $i.e.$, $M\subset H_0(L\otimes I_m)$. From the result in \cite{Saber04}, together with the connectivity of graph $\mathcal{G}$, we have $rank(L)=n-1$. Hence, $dim(H_0(L))=1$, it follows that $dim(H_0(L\otimes I_m))=m=dim(M)$. Thus, $H_0(L\otimes I_m)=M$.
\QEDA

The proof of Lemma \ref{graph} is based on the following two lemmas.
\begin{lemma}(Menger's Theorem \cite{West01})\label{Menger}
If $x$, $y$ are vertices of a graph $\mathcal{G}$ and $(x,y)\notin\mathcal{E}(\mathcal{G})$, then the minimum size of an $x,y-$cut equals the maximum number of pairwise internally disjoint $x,y-$paths.
\end{lemma}
\begin{lemma}(\cite{West01})\label{Deletion}
Deletion of an edge reduces connectivity by at most 1.
\end{lemma}
{\it Proof of Lemma \ref{graph}:} Assume that there exist a pair of agents $i$ and $j$, and the maximum number of disjoint paths between them is $l<k^*$. We discuss the problem in the following two cases.

Case1. If $(i,j)\notin\mathcal{E}(\mathcal{G})$, from Lemma \ref{Menger}, the minimum size of an $i$, $j-$cut in graph $\mathcal{G}$ is $l$. This means that the minimum size of a vertex set disconnecting $i$ and $j$ is $l$. Therefore, $\kappa(\mathcal{G})\leq l<k^*$, which is a contradiction.

Case2. If $(i,j)\in\mathcal{E}(\mathcal{G})$. Let $\mathcal{G}'=\mathcal{G}-\{(i,j)\}$, from Lemma \ref{Deletion}, $\kappa(\mathcal{G}')\geq \kappa(\mathcal{G})-1$. By Menger's Theorem, the minimum size of an $i$, $j-$cut in graph $\mathcal{G}'$ is $l-1$. Hence, $\kappa(\mathcal{G}')\leq l-1$. Then, $\kappa(\mathcal{G})\leq\kappa(\mathcal{G}')+1 \leq l<k^*$, which conflicts with $\kappa(\mathcal{G})=k^*$.
\QEDA
{\it Proof of Lemma \ref{n-1}:} Without loss of generality, suppose that $\mathcal{G}$ has $r$ connected components, with $\mathcal{V}_1,\cdots,\mathcal{V}_r$ as the corresponding set of nodes, $|\mathcal{V}_1|\leq\cdots\leq|\mathcal{V}_r|$. Let $\mathcal{V}_p$ be the first set which has more than one element. That is, $p=\min\limits_{|\mathcal{V}_i|\geq2}\{1,\cdots,r\}$. Let $f(r)$ denote the minimal number of pairs of disconnected nodes, $n_i=|\mathcal{V}_i|$. We have
$$
f(r)= C_n^2-C_{n_p}^2-C_{n_{p+1}}^2-\cdots-C_{n_r}^2, i,j=1,\cdots,r.
$$
Combining $\mathcal{V}_p$ and $\mathcal{V}_{p+1}$, it follows that
$$
f(r-1)\leq C_n^2-C_{n_p+n_{p+1}}^2-C_{n_{p+2}}^2-\cdots-C_{n_r}^2.
$$
Thus,
\begin{center}
$f(r)-f(r-1)\geq C_{n_p+n_{p+1}}^2-C_{n_p}^2-C_{n_{p+1}}^2>0$.
\end{center}
Consequently, $f(r)$ is a decreasing function of $r$. Since the graph is not connected, one has $r>1$. Thus, $f(r)\geq f(2)$. Recalling that $f(2)=\min\{n_1n_2\}=\min\{n_1(n-n_1)\}=n-1$. Therefore, $f(r)\geq n-1$.
\QEDA

{\it Proof of Lemma \ref{bound}:} If $||p(t)||$ is upper bounded, we obtain the upper bound $B$ of $||x_i-x_j||$ for any $i,j\in\mathcal{V}$ from Lemma \ref{||x||}. Using $e$ to denote the eigenvector associated with $\lambda_2(L_x)$, due to the fact that $\alpha(s)$ is nonincreasing of $s$, we have
\[
\begin{split}
\lambda_2(L_x)=\frac{e^TL_xe}{e^Te}&=\frac{\sum\limits_{i\in\mathcal{V}} \sum\limits_{j\in\mathcal{V}}G_{ij}\alpha_{ij}||e_i-e_j||^2}{2e^Te}\\
&\geq\alpha(B)\cdot\frac{\sum\limits_{i\in\mathcal{V}} \sum\limits_{j\in\mathcal{V}}G_{ij}||e_i-e_j||^2}{2e^Te}\\
&=\alpha(B)\cdot\frac{e^T\bar{L}e}{e^Te}\\
&\geq\alpha(B)\lambda_2(\bar{L}),
\end{split}
\]
where $\bar{L}$ is the Laplacian matrix of graph $\bar{\mathcal{G}}=(\mathcal{V}, \mathcal{E}, \bar{\mathcal{A}})$ with $\bar{\mathcal{A}}=G$. Since the communication topology is connected, $\lambda_2(\bar{L})$ is positive constant. Thus, $\lambda_2(L_x)$ has a positive lower bound.
\QEDA

{\it Proof of Proposition \ref{prW}:} (1). It is easy to see that $w(z)=0$ if and only if $z=0$, and $W\geq0$ for any $x\in\mathbb{R}^n$. Then $W=0$ if $x_i=x_j$ for any $i,j\in\mathcal{V}$. Otherwise, suppose $W=0$ is valid, then for any $(i,j)\in\mathcal{E}$, one has $w(||x_{ij}(t)||^2)=W_{ij}=0$, implying that $||x_i-x_j||=0$. Since graph $\mathcal{G}$ is connected, one has $||x_i-x_j||=0$ for any $i,j\in\mathcal{V}$.

(2). Suppose that $0<z_1<z_2$. We study this problem in the following three cases.

Case1. $z_1<z_2<r$. Then $w(z_2)-w(z_1)=\alpha(r)(z_2-z_1)\geq0$.

Case2. $z_1<r\leq z_2$. Then $w(z_2)-w(z_1)\geq\alpha(r)r-\alpha(r)z_1 \geq 0$.

Case3. $r<z_1<z_2$. If $\lfloor\frac{z_2}{r}\rfloor> \lfloor\frac{z_1}{r}\rfloor$, then
\[
\begin{split}
w(z_2)-w(z_1)&\geq\alpha(\lfloor\frac{z_1}{r}\rfloor r+r)r+ \alpha(\lceil \frac{z_2}{r}\rceil r)(z_2-\lfloor\frac{z_2}{r}\rfloor r)\\
&-\alpha(\lceil\frac{z_1}{r}\rceil r)(z_1-\lfloor\frac{z_1}{r}\rfloor r)\\
&\geq\alpha(\lceil\frac{z_2}{r}\rceil r)(z_2-\lfloor\frac{z_2}{r}\rfloor r)\geq0.
\end{split}
\]
If $\lfloor\frac{z_2}{r}\rfloor=\lfloor\frac{z_1}{r}\rfloor$, then $\lceil\frac{z_2}{r}\rceil r=\lceil\frac{z_1}{r}\rceil$ and $z_2-\lfloor\frac{z_2}{r}\rfloor r>z_1-\lfloor\frac{z_1}{r}\rfloor r$. Hence, $w(z_2)-w(z_1)=\alpha(\lceil\frac{z_2}{r}\rceil r)(z_2-\lfloor\frac{z_2}{r}\rfloor r)-\alpha(\lceil\frac{z_1}{r}\rceil r)(z_1-\lfloor\frac{z_1}{r}\rfloor r)\geq0$.

(3). For any $t\geq0$, we discuss the problem in the following two cases.

Case1. $||x_{ij}(t+1)||\geq||x_{ij}(t)||$. Then $\alpha_{ij}(x(t+1))\leq \alpha_{ij}(x(t))$. From (2), one has $W_{ij}(t+1)-W_{ij}(t)\geq0$. Therefore,
$$
W_{ij}(t+1)-W_{ij}(t) \leq \alpha_{ij}(x(t))(||x_{ij}(t+1)||^2-||x_{ij}(t)||^2).
$$
Together with $W(t)=\frac{1}{2}\sum\limits_{i\in\mathcal{V}} \sum\limits_{j\in\mathcal{V}}G_{ij}W_{ij}(t)$, (\ref{dt1W}) is obtained.

Case2. $||x_{ij}(t+1)||<||x_{ij}(t)||$. Then $\alpha_{ij}(x(t+1))\geq \alpha_{ij}(x(t))$. From (2), one has $W_{ij}(t+1)-W_{ij}(t)\leq0$. Therefore,
$$
W_{ij}(t+1)-W_{ij}(t) \leq \alpha_{ij}(x(t))(||x_{ij}(t+1)||^2-||x_{ij}(t)||^2).
$$
Together with $W(t)=\frac{1}{2}\sum\limits_{i\in\mathcal{V}} \sum\limits_{j\in\mathcal{V}}G_{ij}W_{ij}(t)$, (\ref{dt1W}) is obtained.

(4). For any $r\leq z<\infty$, $\alpha(\cdot)$ is Riemann integral on $[0,z]$ since it is monotonous and bounded by $\alpha(0)$. Then we have
\[
\begin{split}
\int_0^z&\alpha(s)ds=\int_0^r\alpha(s)ds+\cdots+\int_{(\lfloor \frac{z}{r}\rfloor-1)r}^{\lfloor \frac{z}{r}\rfloor r}\alpha(s)ds+ \int_{\lfloor \frac{z}{r}\rfloor r}^z\alpha(s)ds\\
&\geq\int_0^r\alpha(r)ds+\cdots+\int_{(\lfloor \frac{z}{r}\rfloor-1)r}^{\lfloor \frac{z}{r}\rfloor r}\alpha(\lfloor \frac{z}{r}\rfloor r)ds+ \int_{\lfloor \frac{z}{r}\rfloor r}^z\alpha(\lceil \frac{z}{r}\rceil r)ds\\
&=\sum\limits_{s=1}^{\lfloor \frac{z}{r}\rfloor}\alpha(sr)+ \alpha(\lceil \frac{z}{r}\rceil r)(z-\lfloor \frac{z}{r}\rfloor r)=w(z).
\end{split}
\]
Furthermore,
\[
\begin{split}
w(z)-\int_r^z\alpha(s)ds&= \sum\limits_{s=1}^{\lfloor\frac{z}{r}\rfloor}\alpha(sr)r+ \alpha(\lceil\frac{z}{r}\rceil r)(z-\lfloor\frac{z}{r}\rfloor r)-\int_r^{\lfloor\frac{z}{r}\rfloor r}\alpha(s)ds -\int_{\lfloor\frac{z}{r}\rfloor r}^z\alpha(s)ds\\
&\geq\alpha(\lfloor\frac{z}{r}\rfloor r)r+ \alpha(\lceil\frac{z}{r}\rceil r)(z-\lfloor\frac{z}{r}\rfloor r)- \int_{\lfloor\frac{z}{r}\rfloor r}^z\alpha(s)ds\\
&\geq\alpha(\lfloor\frac{z}{r}\rfloor r)(z-\lfloor\frac{z}{r}\rfloor r)\geq0.
\end{split}
\]
Therefore, it holds that $\int_r^z\alpha(s)ds\leq w(z)\leq\int_0^z\alpha(s)ds$. Since $\lim\limits_{r\rightarrow0} \int_r^z\alpha(s)ds=\int_0^z\alpha(s)ds$, together with Squeeze Theorem, it follows that $\lim\limits_{r\rightarrow0}w(z) =\int_0^z\alpha(s)ds$ for $z\geq0$.
\QEDA

The proof of Proposition \ref{pr 2n-3} is based on the following Lemma.
\begin{lemma}\label{pr symmetry}
If the initial states are symmetrically distributed, the states of all the agents in model (\ref{smooth}) will be symmetrically distributed for any $t\geq0$.
\end{lemma}
\begin{proof} Suppose that all the opinions are symmetrically distributed at time $t\geq0$. For any $i+j=n+1$($i$, $j$ can be the same), the symmetric distribution implies that $x_i(t)+x_j(t)=x_1(t)+x_n(t)$, the neighbors of $i$ and $j$ are also symmetrically distributed. That is, for any $k\in\mathcal{N}_i(t)$, there exists a unique $l\in\mathcal{N}_j(t)$, such that $k+l=n+1$. Moreover, since $x_i(t)+x_j(t)=x_k(t)+x_l(t)$, one has $x_i(t)-x_k(t)=x_l(t)-x_j(t)$, implying that $\alpha_{ik}=\alpha_{jl}$. Therefore,
\[
\begin{split}
\dot{x}_i(t)+\dot{x}_j(t)&=\sum_{k\in\mathcal{N}_i(t)} \alpha_{ik}(x_k(t)-x_i(t))+ \sum_{l\in\mathcal{N}_j(t)} \alpha_{jl}(x_l(t)-x_j(t))\\
&=\sum_{k\in\mathcal{N}_i(t)}\alpha_{ik}x_k(t)+ \sum_{l\in\mathcal{N}_j(t)} \alpha_{jl}x_l(t)\\
&- \sum_{k\in\mathcal{N}_i(t)} \alpha_{ik}x_i(t)-\sum_{l\in\mathcal{N}_j(t)} \alpha_{lj}x_j(t)=0.
\end{split}
\]
Hence, $M=\{x~|~x_i(t)+x_j(t)=x_1(t)+x_n(t),~ i+j=n+1\}$ is a positively invariant set. Since $x(0)\in M$, the states will always be symmetrically distributed.
\end{proof}

{\it Proof of Proposition \ref{pr symmetry}:} Suppose graph $G$ has $r$ connected components with $\mathcal{V}_1$, $\cdots$, $\mathcal{V}_r$ as their vertex sets, and $|\mathcal{V}_i|=n_i$ for $i\in\{1,\cdots,r\}$. Since the agents are always symmetrically distributed, we let $n_k=n_j$ for any $k+j=1+r$. Let $g(r)$ be the number of pairs of connected nodes. Then $g(r)=\sum_{n_i>1}C_{n_i}^2$. We consider the problem in the following two cases:

Case1, n is odd. From the symmetry, we have $1\leq n_1\leq\frac{n-1}{2}$.

If $n_1=1$, $g(r)\leq C_{\sum_{n_i>1}n_i}^2\leq C_{n-n_1-n_r}^2=C_{n-2}^2 =\frac{n^2-5n+6}{2}$.

If $n_1=\frac{n-1}{2}$, $g(r)=C_{n_1}^2+C_{n_r}^2=\frac{n^2-4n+3}{4}$.

If $1<n_1<\frac{n-1}{2}$, $g(r)\leq C_{n_1}^2+C_{n-n_1-n_r}^2+C_{n_r}^2 =3n_1^2-2nn_1+\frac{n^2-n}{2}\leq\frac{n^2-8n+27}{4}$.

Case2, n is even. From the symmetry, we have $1\leq n_1\leq\frac{n}{2}$.

If $n_1=1$, $g(r)\leq C_{\sum_{n_i>1}n_i}^2\leq C_{n-n_1-n_r}^2=C_{n-2}^2 =\frac{n^2-5n+6}{2}$.

If $n_1=\frac{n}{2}$, $g(r)=C_{n_1}^2+C_{n_r}^2=\frac{n^2-2n}{4}$.

If $1<n_1<\frac{n}{2}$, $g(r)\leq C_{n_1}^2+C_{n-n_1-n_r}^2+C_{n_r}^2 =3n_1^2-2nn_1+\frac{n^2-n}{2}\leq\frac{n^2-6n+12}{4}$.

In conclusion, we can obtain that $g(r)\leq\frac{n^2-5n+6}{2}$ for $n\geq4$. Therefore, the minimal number of pairs of disconnected agents is $f(r)=C_n^2-g(r)\geq2n-3$.
\QEDA



\ifCLASSOPTIONcaptionsoff
  \newpage
\fi



%

\end{document}